\pdfoutput=1
\documentclass[final]{biometrika}

\usepackage{amsmath}
\usepackage{amssymb}
\usepackage{bm}
\usepackage{bbm}
\usepackage{graphicx}
\usepackage{subfigure}
\usepackage[plain,noend]{algorithm2e}
\usepackage{placeins} 
\usepackage{natbib}
\usepackage{enumitem} 
\usepackage[title]{appendix} 
\usepackage{xcolor} 
\usepackage{verbatim} 
\usepackage{times}

\makeatletter
\renewcommand{\algocf@captiontext}[2]{#1\algocf@typo. \AlCapFnt{}#2} 
\def\@algocf@capt@plain{top}
\renewcommand{\algocf@makecaption}[2]{%
  \addtolength{\hsize}{\algomargin}%
  \sbox\@tempboxa{\algocf@captiontext{#1}{#2}}%
  \ifdim\wd\@tempboxa >\hsize
    \hskip .5\algomargin%
    \parbox[t]{\hsize}{\algocf@captiontext{#1}{#2}}
  \else%
    \global\@minipagefalse%
    \hbox to\hsize{\box\@tempboxa}
  \fi%
  \addtolength{\hsize}{-\algomargin}%
}
\makeatother

\setlist[enumerate]{itemsep=2pt, topsep=4pt}

\newcommand{\sign}{\textrm{sign}}
\newcommand{\hess}{\bm{\mathcal{H}}}
\newcommand{\hessi}{\bm{\mathcal{I}}}
\newcommand{\half}{\scalebox{0.7}{1/2}}

\newcommand{\minus}{\scalebox{1}{-}}
\newcommand{\pot}{U}
\newcommand{\R}{\bm{R}}
\newcommand{\bPsi}{\bm{\Psi}}
\newcommand{\I}{\bm{I}}
\newcommand{\btheta}{\bm{\theta}}
\newcommand{\p}{\bm{p}}
\newcommand{\e}{\bm{e}}
\newcommand{\bnu}{\bm{\nu}}
\newcommand{\by}{\bm{y}}
\newcommand{\bv}{\bm{v}}
\newcommand{\bw}{\bm{w}}

\newcommand{\tilN}{\widetilde{N}}
\newcommand{\tileps}{\tilde{\epsilon}}
\newcommand{\normal}{\mathcal{N}}
\newcommand{\mass}{\bm{M}}
\newcommand{\diff}{\textrm{d}}
\newcommand{\given}{\, | \,}

\newcommand{\pushright}[1]{\ifmeasuring@#1\else\omit\hfill$\displaystyle#1$\fi\ignorespaces}
\newcommand{\nobs}{y}
\newcommand{\nunknown}{N}
\newcommand{\hmc}{Hamiltonian Monte Carlo}
\newcommand{\dhmc}{discontinuous Hamiltonian Monte Carlo}
\newcommand{\mcmc}{Markov chain Monte Carlo}
\newcommand{\ess}{effective sample size}
\newcommand{\nuts}{no-U-turn sampler}
\newcommand{\nutsGibbs}{No-U-turn / Gibbs}
\newcommand{\nobold}{0}
\if\nobold1
\renewcommand{\bm}[1]{#1}
\fi

\begin{document}

\jname{}
\jyear{}
\jvol{}
\jnum{}
\accessdate{}


\markboth{Nishimura et~al.}{Discontinuous \hmc{}}

\title{Discontinuous Hamiltonian Monte Carlo \\ for discrete parameters and discontinuous likelihoods}

\author{Akihiko Nishimura}
\affil{Department of Biomathematics, University of California - Los Angeles, \\695 Charles E Young Dr S, Los Angeles, California 90095, U.S.A. \email{akihiko4@ucla.edu}}

\author{David B.\ Dunson}
\affil{Department of Statistical Science, Duke University, \\Box 90251, Durham, North Carolina 27708, U.S.A. \email{dunson@duke.edu}}

\author{Jianfeng Lu}
\affil{Department of Mathematics, Duke University, \\Box 90320, Durham, North Carolina 27708, U.S.A. \email{jianfeng@math.duke.edu}}

\maketitle

\begin{abstract}
Hamiltonian Monte Carlo has emerged as a standard tool for posterior computation. In this article, we present an extension that can efficiently explore target distributions with discontinuous densities. Our extension in particular enables efficient sampling from ordinal parameters though embedding of probability mass functions into continuous spaces. We motivate our approach through a theory of discontinuous Hamiltonian dynamics and develop a corresponding numerical solver. The proposed solver is the first of its kind, with a remarkable ability to \textit{exactly} preserve the Hamiltonian. 
We apply our algorithm to challenging posterior inference problems to demonstrate its wide applicability and competitive performance. 
\end{abstract}


\section{Introduction}
\label{sec:intro}

Markov chain Monte Carlo is routinely used to generate samples from posterior distributions. While specialized algorithms exist for restricted model classes, general-purpose algorithms are often inefficient and scale poorly in the number of parameters. Originally proposed by  \citet{duane87} and popularized in the statistics community through the works of \citet{neal96, neal10}, Hamiltonian Monte Carlo promises a better scalability \citep{neal10, beskos13} and has enjoyed wide-ranging successes as one of the most reliable approaches in general settings \citep{bda13, kruschke14, monnahan16}.  

However, a fundamental limitation of \hmc{} is the lack of support for discrete parameters \citep{gelman15, monnahan16}. 
The difficulty stems from the fact that the construction of \hmc{} proposals relies on a numerical solution of a differential equation. The use of a surrogate differentiable target distribution may be possible in special cases \citep{zhang12}, but approximating a discrete parameter of a likelihood by a continuous one is difficult in general \citep{berger12}.

This article presents \textit{discontinuous \hmc{}}, an extension that can efficiently explore spaces involving ordinal discrete parameters as well as continuous ones. 
The algorithm can also handle discontinuities in piecewise smooth posterior densities, which for example arise from models with structural change points \citep{chib1998changepoint,wagner2002segmented}, latent thresholds \citep{neelon2004isotonic, nakajima2013threshold}, and pseudo-likelihoods \citep{bissiri16}. 

\expandafter\MakeUppercase \dhmc{} retains the generality that makes \hmc{} suitable for automatic posterior inference. For any given target distribution, each iteration only requires evaluations of the density and of the following quantities up to normalizing constants: 1) full conditional densities of discrete parameters, and 2) either the gradient of the log density with respect to continuous parameters or their individual full conditional densities. Evaluations of full conditionals can be done algorithmically and efficiently through directed acyclic graph frameworks, taking advantage of conditional independence structures \citep{lunn2009bugs}. Algorithmic evaluation of the gradient is also efficient \citep{griewank2008algodiff} and its implementations are widely available as open-source modules \citep{carpenter2015stan}. 

In our framework, the discrete parameters are first embedded into a continuous space, inducing parameters with piecewise constant densities. A key theoretical insight is that Hamiltonian dynamics with a discontinuous potential energy can be integrated analytically near its discontinuity in a way that exactly preserves the total energy. This fact was realized by \citet{pakman13} and used to sample from binary distributions through embedding them into a continuous space. This framework was extended by \citet{afshar15} to handle more general discontinuous distributions and then by \citet{dinh2017probabilistic-hmc} to settings where the parameter space involves phylogenetic trees. All these frameworks, however, run into serious computational issues when dealing with more complex discontinuities and thus fail as general-purpose algorithms.

We introduce novel techniques to obtain a practical sampling algorithm for discrete parameters and, more generally, target distributions with discontinuous densities. In dealing with discontinuous targets, we propose a Laplace distribution for the momentum variable as a more effective alternative to the usual Gaussian distribution. We develop an efficient integrator of the resulting Hamiltonian dynamics by splitting the differential operator into its coordinate-wise  components.
A version of \dhmc{} coincides with a generalization of Metropolis-within-Gibbs samplers, overcoming dependency among the parameters by adding momentum along each coordinate.

\section{\hmc{} for discrete and discontinuous distributions}
\label{sec:event_driven_hmc}

\subsection{Review of \hmc{}}
Given a parameter $\btheta \sim \pi_{\Theta}(\cdot)$ of interest, \hmc{} introduces an auxiliary \textit{momentum} variable $\p$ and samples from a joint distribution $\pi(\btheta, \p) = \pi_{\Theta}(\btheta) \pi_P(\p)$ for some symmetric distribution $\pi_P(\p) \propto \exp\{- K(\p)\}$. 
The function $K(\p)$ is referred to as the \textit{kinetic energy} and $\pot(\btheta) = - \log \pi_{\Theta}(\btheta)$ as the \textit{potential energy}. The total energy $H(\btheta, \p) = U(\btheta) + K(\p)$ is often called the \textit{Hamiltonian}. A proposal is generated by simulating trajectories of \textit{Hamiltonian dynamics} where the evolution of the state $(\btheta, \p)$ is governed by \textit{Hamilton's equations}:
\begin{equation} 
\label{eq:hamilton}
\begin{aligned}
\frac{\diff \btheta}{\diff t}
&= \nabla_{\p} K(\p), \quad 
\frac{\diff \p}{\diff t}
= - \nabla_{\btheta} U(\btheta) = \nabla_{\btheta} \log \pi_{\Theta}(\btheta).
\end{aligned}
\end{equation}
The next section shows how we can turn the problem of dealing with a discrete parameter $\btheta$ to that of dealing with a discontinuous target density. We then proceed to make sense of the differential equation \eqref{eq:hamilton} when $\pi_{\Theta}(\btheta)$, and hence $U(\btheta)$, is discontinuous.

\subsection{Dealing with discrete parameters via embedding}
\label{sec:embedding}
Let $N$ denote a discrete parameter with the prior distribution $\pi_N(\cdot)$ and assume without loss of generality that $N$ takes positive integer values. 
For example, the inference goal may be estimation of the population size $N$ given the observation $y \given q, N \sim \mathrm{Binomial}(q, N)$. 
We embed $N$ into a continuous space by introducing a latent parameter $\tilN$ whose relationship to $N$ is defined as
\begin{equation}
N = n \quad \text{ if and only if } \quad \tilN \in (a_n, a_{n+1}],
\end{equation}
for an increasing sequence of real numbers $0 = a_1 \leq a_2 \leq a_3 \leq \ldots$. To match the prior distribution on $N$, we define the corresponding prior density on $\tilN$ as
\begin{equation}
\pi_{\tilN}(\tilde{n})
= \sum_{n \geq 1} \frac{\pi_N(n)}{a_{n+1} - a_n} \, \mathbbm{1}\{ a_n < \tilde{n} \leq a_{n+1} \},
\end{equation}
where the Jacobian-like factor $(a_{n+1} - a_n)^{-1}$ adjusts for embedding into non-uniform intervals. 

Although the choice $a_n = n$ for all $n$ is a natural one, a non-uniform embedding is useful in effectively carrying out a parameter transformation of $N$. For example, a log-transform $a_n = \log n$ may be used to avoid a heavy-tailed distribution on $\tilN$ or to reduce correlation between $\tilN$ and the rest of the parameters. Mixing of many \mcmc{} algorithms, including \dhmc{}, can be substantially improved by such parameter transformations \citep{roberts2009adap-mcmc, thawornwattana2018designing-mcmc-proposal}.

While the above strategy can be applied whether or not the discrete parameter values have a natural ordering, embedding the values in an arbitrary order likely induces a multi-modal continuous distribution. 
The mixing rate of (discontinuous) \hmc{} generally suffers from multi-modality due to the energy-conservation property of the dynamics \citep{neal10}.

\subsection{How \hmc{} fails on discontinuous target densities}
\label{sec:hmc_fails_on_discontinuity}
Having recast the discrete parameter problem as a discontinuous one, we focus the rest of our discussion on discontinuous targets. 
An \textit{integrator} is an algorithm that numerically approximates an evolution of the exact solution to a differential equation. \hmc{} requires \textit{reversible} and \textit{volume-preserving} integrators to guarantee symmetry of its proposal distributions (see Section~\ref{sec:dhmc_reversibility} and \citealp{neal10}). 
These proposals are generated as follows: 
\begin{enumerate}
	\item Sample the momentum from its marginal distribution $\p \sim \pi_P(\cdot)$.
	\item Using a reversible and volume-preserving integrator, approximate $\{\btheta(t), \p(t)\}_{t \geq 0}$, the solution of the differential equation \eqref{eq:hamilton} with the initial condition $\{\btheta(0), \p(0)\} = (\btheta, \p)$. Use the approximate solution  $(\btheta^*, \p^*) \approx \{\btheta(\tau), \p(\tau)\}$ for some $\tau > 0$ as a proposal. 
\end{enumerate}
The proposal $(\btheta^*, \p^*)$ then is accepted with Metropolis probability \citep{metropolis53}
\begin{equation}
\min \left[ 1, \ \exp\{ - H(\btheta^*, \p^*) + H(\btheta, \p)\} \right],
\end{equation}
where $H(\btheta, \p) = - \log \pi(\btheta, \p)$ denotes the Hamiltonian.
With an accurate integrator, the acceptance probability of
$(\btheta^*, \p^*)$ can be close to $1$ because the exact dynamics conserves the energy:
$H\{\btheta(t), \p(t)\}= H\{\btheta(0), \p(0)\}$ for all $t \geq 0$. The integrator of choice for \hmc{} is the
\textit{leapfrog} scheme, which approximates the evolution
$\{\btheta(t), \p(t)\} \to \{\btheta(t + \epsilon), \p(t + \epsilon)\}$ by the successive updates
\begin{equation}
\label{eq:leapfrog}
\begin{aligned}
\p &\gets \p - \frac{\epsilon}{2} \nabla_{\btheta} \pot(\btheta), \quad
\btheta	\gets \btheta + \epsilon \nabla_{\p} K(\p), \quad
\p \gets \p - \frac{\epsilon}{2} \nabla_{\btheta} \pot(\btheta).
\end{aligned}
\end{equation}
When $\pi_\Theta(\cdot)$ is smooth, approximating the evolution $\{\btheta(0), \p(0)\} \to \{\btheta(\tau), \p(\tau)\}$ with $L = \lfloor \tau / \epsilon \rfloor$ leapfrog steps results in a global error of order $O(\epsilon^2)$ so that $H(\btheta^*, \p^*) = H(\btheta, \p) + O(\epsilon^2)$ \citep{hairer06}. \hmc{}'s high acceptance rates and scaling properties critically depend on this second-order accuracy \citep{beskos13}. When $\pi_\Theta(\cdot)$ has a discontinuity, however, the leapfrog updates \eqref{eq:leapfrog} fail to account for the instantaneous change in $\pi_\Theta(\cdot)$, incurring an unbounded error that does not decrease even as $\epsilon \to 0$ (supplement Section~\ref{sec:leapfrog_on_discontinuous_target}).

\subsection{Theory of discontinuous Hamiltonian dynamics}
\label{sec:event_driven_approach}

Suppose $U(\theta)$ is piecewise smooth; that is, the domain has a partition $\mathbb{R}^d = \cup_k \Omega_k$ such that $U(\theta)$ is smooth on the interior of $\Omega_k$ and the boundary $\partial \Omega_k$ is a $(d - 1)$-dimensional piecewise smooth manifold.
While the classical definition of gradient $\nabla U(\btheta)$ makes no sense on a discontinuity set $\cup_k \partial \Omega_k$, it can be defined through a notion of \textit{distributional derivatives} and the corresponding Hamilton's equations \eqref{eq:hamilton} can be interpreted as a \textit{measure-valued differential inclusion} \citep{stewart00}. In this case, a solution is in general non-unique unlike that of a smooth ordinary differential equation.  To find a solution that preserves critical properties of Hamiltonian dynamics, we rely on a so-called \textit{selection principle} \citep{ambrosio08} as follows.

Define a sequence of smooth approximations $\pot_\delta(\btheta)$ of
$\pot(\btheta)$ for $\delta > 0$ through the convolution
$\pot_\delta(\theta) := \int \pot(\eta) \phi_\delta(\theta - \eta) \diff \eta$, where
$\phi_\delta(\btheta) = \delta^{-d} \phi(\delta^{-1} \btheta)$ is a
compactly supported smooth function with $\phi \geq 0$ and
$\int \phi = 1$. 
Now let $\{\btheta_{\delta}(t), \p_{\delta}(t)\}_{t \geq 0}$ be the
solution of Hamilton's equations with the potential energy
$\pot_\delta$. The pointwise limit
$\{\btheta(t), \p(t)\} = \lim_{\delta \to 0} \{\btheta_{\delta}(t),
\p_{\delta}(t)\}$
can be shown to exist for almost every initial condition on some time interval $[0, T]$  \citep{hirsch1974ode}. The collections of the trajectories $\left\{ \{\btheta(t), \p(t)\}_{t \geq 0} :  \{\btheta(0), \p(0)\} \in \mathbb{R}^d \right\}$ then define the dynamics
corresponding to $\pot(\btheta)$. This construction provides a rigorous mathematical foundation for the special cases of discontinuous Hamiltonian dynamics derived by
\citet{pakman13} and \citet{afshar15} through physical heuristics. 

The behavior of the limiting dynamics near the discontinuity is
deduced as follows. Suppose that the trajectory $\{\btheta(t), \p(t)\}$
hits the discontinuity at an event time $t_e$ and let $t_e^-$ and
$t_e^+$ denote infinitesimal moments before and after that. At a discontinuity point $\btheta \in \partial \Omega_k$, we have
\begin{equation}
\label{eq:bdry_regularity}
\lim_{\delta \to 0} \nabla_{\btheta} \pot_\delta(\btheta) / \| \nabla_{\btheta} \pot_\delta(\btheta) \| =  \bnu(\btheta),
\end{equation}
where $\bnu(\btheta)$ denotes a unit vector orthogonal to $\partial \Omega_k$, pointing in the direction of higher potential energy. 
The relations \eqref{eq:bdry_regularity} and $\diff \p_{\delta} / \diff t = - \nabla_{\btheta} \pot_\delta$ imply that the only change in $\p(t)$ upon encountering the discontinuity occurs in the direction of $\bnu_e = \bnu\{\btheta(t_e)\}$:
\vspace{-.25\baselineskip}
\begin{equation}
\label{eq:orthogonal_change_at_discontinuity}
\p(t_e^+) = \p(t_e^-)  - \gamma \, \bnu_e 
\end{equation}
for some $\gamma > 0$. There are two possible types of change in $\p$ depending on the potential energy difference $\Delta \pot_e$ at the discontinuity, which we formally define as
\begin{equation}
\label{eq:energy_diff}
\Delta \pot_e = \lim_{\epsilon \to 0^+} \pot\!\left\{\btheta(t_e^+) + \epsilon \p(t_e^-)\right\} - \pot\!\left\{\btheta(t_e^-)\right\}.
\end{equation}
When the momentum does not provide enough kinetic energy to overcome the potential energy increase $\Delta U_e$, the trajectory bounces against the plane orthogonal to $\bnu_e$. Otherwise, the trajectory moves through the discontinuity by transferring kinetic energy to potential energy. Either way, the magnitude of an instantaneous change $\gamma$ can be determined via the energy conservation law:
\begin{equation}
\label{eq:energy_conservation_at_discontinuity}
K\{\p(t_e^+)\} - K\{\p(t_e^-)\} = \pot\{\btheta(t_e^-)\} - \pot\{\btheta(t_e^+)\}.
\end{equation}
Figure~\ref{fig:dhmc_visualization}, which is explained in more detail in Section~\ref{sec:dhmc}, provides a visual illustration of the trajectory behavior at a discontinuity.

\section{Integrator for discontinuous dynamics via Laplace momentum}
\label{sec:dhmc}

\subsection{Limitation of Gaussian momentum-based approaches}
\label{sec:issue_with_gaussian_mom}
Use of non-Gaussian momentums has received limited attention in the \hmc{} literature. Correspondingly, the existing discontinuous extensions all rely on Gaussian momentums \citep{pakman13, afshar15, dinh2017probabilistic-hmc}. 

In developing a general-purpose algorithm for sampling from discontinuous targets, however, dynamics based on a Gaussian momentum have a serious shortcoming.  In order to approximate the dynamics accurately, the integrator must detect every single discontinuity encountered by a trajectory and then compute the potential energy difference each time (Algorithm~\ref{alg:ed_integrator_with_gaussian_mom} in the supplement Section~\ref{sec:ed_integrator_for_gaussian_mom}). To see why this is a serious problem, consider a discrete parameter $N \in \mathbb{Z}^+$ with a substantial posterior uncertainty, say $\textrm{Var}(N \given \text{data})^{1/2} \approx 1000$. We can then expect a Metropolis move like $N \to N \pm 1000$ to be accepted with a moderate probability, costing only a single likelihood evaluation. On the other hand, if we were to sample a continuously embedded counter part $\tilN$ of $N$ using \dhmc{} with the Gaussian momentum-based Algorithm~\ref{alg:ed_integrator_with_gaussian_mom}, a transition of the corresponding magnitude necessarily requires \emph{1000 likelihood evaluations}. The algorithm is made practically useless by such a high computational cost for otherwise simple parameter updates.

\subsection{Hamiltonian dynamics based on Laplace momentum}
\label{sec:dynamics_with_laplace_mom}
The above example exposes a major challenge in turning discontinuous Hamiltonian dynamics into a practical general-purpose sampling algorithm: an integrator must rely only on a small number of target density evaluation to jump through multiple discontinuities while approximately preserving the total energy. 
We employ a Laplace momentum $\pi_P(\p) \propto \prod_i \exp(- m_i^{-1} | p_i |)$ to provide a solution. While similar distributions have been considered for improving numerical stability of traditional \hmc{} \citep{zhang16, lu16, livingstone2019kinetic-energy-choice}, we exploit a unique feature of the Laplace momentum in a fundamentally new way.

Hamilton's equation under the independent Laplace momentum is given by
\begin{equation}
\label{eq:hamilton_for_laplace}
\frac{\diff \btheta}{\diff t}
= \bm{m}^{-1} \odot \textrm{sign}(\p), \quad 
\frac{\diff \p}{\diff t}
= - \nabla_{\btheta} \pot(\btheta),
\end{equation}
where $\odot$ denotes element-wise multiplication. A key characteristic of the dynamics \eqref{eq:hamilton_for_laplace} is that ${\diff \btheta}/{\diff t}$ depends only on the signs of $p_i$'s and not on their magnitudes. In particular, if we know that $p_i(t)$'s do not change their signs on the time interval $t \in [\tau, \tau + \epsilon]$, then we also know that
\begin{equation}
\btheta(\tau + \epsilon) = \btheta(\tau) + \epsilon \, \bm{m}^{-1} \odot \textrm{sign}\{\p(\tau) \}
\end{equation}
\emph{irrespective of the intermediate values} $\pot\{\btheta(t)\}$ along the trajectory $\{\btheta(t), \p(t)\}$ for $t \in$ \mbox{$[\tau, \tau + \epsilon]$}.
Our integrator critically takes advantage of this property so that it can jump through multiple discontinuities of $\pot(\btheta)$ in just a single target density evaluation.

\subsection{Integrator for Laplace momentum via operator splitting}
\label{sec:integrator_for_laplace_mom}

Operator splitting is a technique to approximate the solution of a differential equation by decomposing it into components, each of which can be solved more easily \citep{mclachlan02}. More commonly used Hamiltonian splitting methods are special cases \citep{neal10}. A convenient splitting scheme for \eqref{eq:hamilton_for_laplace} can be devised by considering the equation for each coordinate $(\theta_i, p_i)$ while keeping the other parameters $(\btheta_{\minus i}, \p_{\minus i})$ fixed:
\begin{equation}
\label{eq:coord_hamilton_for_laplace}
\begin{aligned}
&\frac{\diff \theta_i}{\diff t}
= m_i^{-1} \, \textrm{sign}(p_i), \quad
\frac{\diff p_i}{\diff t} = - \partial_{\theta_i} U(\btheta), \quad
\frac{\diff \btheta_{\minus i}}{\diff t} = \frac{\diff \p_{\minus i}}{\diff t} = \bm{0}.
\end{aligned}
\end{equation}
There are two possible behaviors for the solution $\{\btheta(t), \p(t)\}$ of \eqref{eq:coord_hamilton_for_laplace} for $t \in [\tau, \tau + \epsilon]$, depending on the initial momentum $p_i(\tau)$. Let $\btheta^*(t)$ denote a potential path of $\btheta(t)$: 
\begin{equation}
\theta^*_i(t) = \theta_i(\tau) + (t - \tau) m_i^{-1} \textrm{sign}(p_i(\tau)), \quad
\btheta^*_{\minus i}(t) = \btheta_{\minus i}(\tau).
\end{equation}
In case the initial momentum is large enough that $m_i^{-1} \vert p_i(\tau) \vert > \pot\big\{\btheta^*(t) \big\} - \pot\{{\btheta}(\tau)\}$ for all $t \in [\tau, \tau + \epsilon]$, we have
\begin{equation}
\btheta(\tau + \epsilon) 
= \btheta^*(\tau + \epsilon) 
= \btheta(\tau) + \epsilon \, m_i^{-1} \textrm{sign}\{p_i(\tau)\} \e_i.
\end{equation}
Otherwise, the momentum $p_i$ flips ($p_i \gets - p_i$) and the trajectory $\btheta(t)$ reverses its course at the event time $t_e$ given by
\begin{equation}
t_e = \inf \Big\{t \in [\tau, \tau + \epsilon] : \pot\{\btheta^*(t) \} - \pot\{{\btheta}(\tau) \} > K\{\p(\tau)\} \Big\}.
\end{equation}
See Figure~\ref{fig:dhmc_visualization} for a visual illustration of the trajectory $\btheta(t)$. 

\begin{figure}
	\hspace*{-.1\linewidth}
	\begin{minipage}{.48\linewidth}
	\vspace{-\baselineskip}
	\includegraphics[scale=.7]{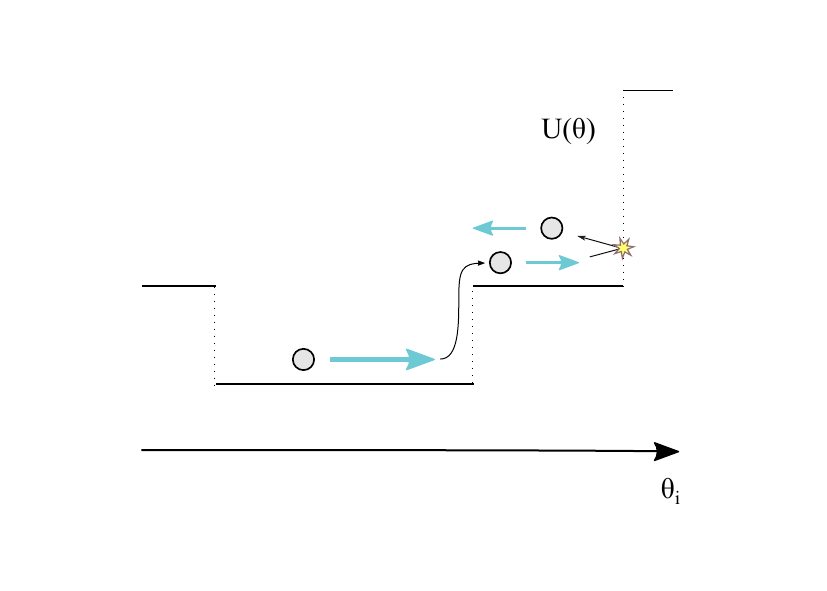}
	\end{minipage}
	\hspace*{-.16\linewidth}
	\begin{minipage}{.48\linewidth}
	\vspace{-2\baselineskip}
	\caption{An example trajectory $\btheta(t)$ of discontinuous Hamiltonian dynamics. The trajectory has enough kinetic energy to move across the first discontinuity by transferring some kinetic energy to potential energy. Across the second discontinuity, however, the trajectory has insufficient kinetic energy to compensate for the potential energy increase and bounces back as a result. }
	\label{fig:dhmc_visualization}
	\end{minipage}
	\vspace{-2.5\baselineskip}
\end{figure}

By emulating the behavior of the solution $\{\btheta(t), \p(t)\}$, we obtain an efficient integrator of the coordinate-wise equation \eqref{eq:coord_hamilton_for_laplace} as given in Algorithm~\ref{alg:coord_integrator_with_laplace_mom}. While the parameter $\btheta$ does not get updated when $m_i^{-1} | p_i | < \Delta \pot$ (line~\ref{line:compare_energy}), the momentum flip $p_i \gets - p_i$ (line~\ref{line:mom_flip}) ensures that the next numerical integration step leads the trajectory toward a higher density of $\pi_\Theta(\btheta)$.
This can be viewed as a generalization of the guided random walk idea by \cite{gustafson1998guided-metropolis}.

The solution of the original (unsplit) differential equation \eqref{eq:hamilton_for_laplace} is approximated by sequentially updating each coordinate of $(\btheta, \p)$ with Algorithm~\ref{alg:coord_integrator_with_laplace_mom} as illustrated in Figure~\ref{fig:2d_traj_of_coord_integrator}. The reversibility of the resulting proposal is guaranteed by randomly permuting the order of the coordinate updates. Alternatively, one can split the operator symmetrically to obtain a reversible integrator \citep{mclachlan02}.
The integrator does not reproduce the exact solution but nonetheless preserves the Hamiltonian exactly. This remains true with any stepsize $\epsilon$, but for good mixing the stepsize needs to be chosen small enough that the condition on Line~\ref{line:compare_energy} is satisfied with high probability (supplement Section~\ref{sec:tuning_dhmc}).

\subsection{Mixing momentum distributions for continuous and discrete parameters}
\label{sec:mixing_momentum_dist}
The potential energy $\pot(\btheta)$ is a smooth function of $\theta_i$ if both the prior and likelihood depend smoothly on $\theta_i$.
The coordinate-wise update of Algorithm~\ref{alg:coord_integrator_with_laplace_mom} leads to a valid proposal mechanism whether or not $\pot(\btheta)$ has discontinuities along $\theta_i$. If $\pot(\btheta)$ varies smoothly along some coordinates of $\btheta$, however, we can devise an integrator that takes advantage of the smooth dependence.

We first write $\btheta = (\btheta_I, \btheta_J)$ where the collections of indices $I$ and $J$ are defined as
\begin{equation}
\label{eq:disc_indiices}
\begin{aligned}
I &= \left\{ i \in \{1, \ldots, d\} : \text{$\pot(\btheta)$ is a smooth function of $\theta_i$} \right\}, \quad J = \{1, \ldots, d\} \setminus I.
\end{aligned}
\end{equation}

\begin{minipage}{0.45\linewidth}
\hspace*{-.05\linewidth}
{\LinesNumbered
	\begin{algorithm}[H]
		\caption{Coordinate-wise integrator for dynamics with Laplace momentum}
		\label{alg:coord_integrator_with_laplace_mom}
		\DontPrintSemicolon
		\SetKwFunction{CoordIntegrator}{CoordIntegrator}
		\SetKwProg{Fn}{Function}{:}{}{}
		\Fn{\mbox{\CoordIntegrator $(\btheta, \p, i, \epsilon)$}}{
			$\btheta^* \gets \btheta$ \;
			$\theta^*_i \gets \theta^*_i + \epsilon m_i^{-1} \textrm{sign}(p_i)$ \;
			$\Delta \pot \gets \pot(\btheta^*) - \pot(\btheta)$ \;
			\eIf{$m_i^{-1} |p_i| > \Delta \pot$ \label{line:compare_energy}}{
				$\theta_i \gets \theta_i^*$ \;
				$p_i \gets p_i - \textrm{sign}(p_i) m_i \Delta \pot$ \;
			}{
				$p_i \gets - p_i$ \; \label{line:mom_flip}
			}
			\Return{$\btheta, \, \p$}
		}
	\end{algorithm}
}
\end{minipage}
\begin{minipage}{0.55\linewidth}
\begin{algorithm}[H]
	\setlength{\lineskip}{4pt}
	\SetInd{0.5em}{0.75em}
	\caption{Integrator for discontinuous \newline \hmc{}}
	\label{alg:dhmc_integrator}
	\DontPrintSemicolon
	\SetKwFunction{DiscIntegrator}{DiscIntegrator}
	\SetKwFunction{Permute}{Permute}
	\SetKwProg{Fn}{Function}{:}{}{}
	\vspace{.3\baselineskip}
	\Fn{\mbox{\DiscIntegrator $\left(\btheta, \p, \epsilon, \varphi = \Permute(J)\right)$}}{
		$\p_I \gets \p_I + \dfrac{\epsilon}{2} \nabla_{\btheta_I} \log \pi(\btheta)$ \;
		$\btheta_I \gets \btheta_I + \dfrac{\epsilon}{2} \, \mass_I^{-1} \p_I	$ \;
		\For{$j \ \mathrm{in} \ J$}{
			\mbox{$\btheta, \, \p \gets \CoordIntegrator(\btheta, \p, \varphi(j), \epsilon)$}
			\tcp{\footnotesize \hspace{-.5ex}Update discontinuous params}
		}
		$\btheta_I \gets \btheta_I + \dfrac{\epsilon}{2} \, \mass_I^{-1} \p_I	$ \;
		$\p_I \gets \p_I + \dfrac{\epsilon}{2} \nabla_{\btheta_I} \log \pi(\btheta)$ \;
		\Return{$\btheta, \, \p$}
	}
\end{algorithm}
\end{minipage}

\begin{figure}
	\begin{minipage}{.5\linewidth}
	\includegraphics[width=\linewidth]{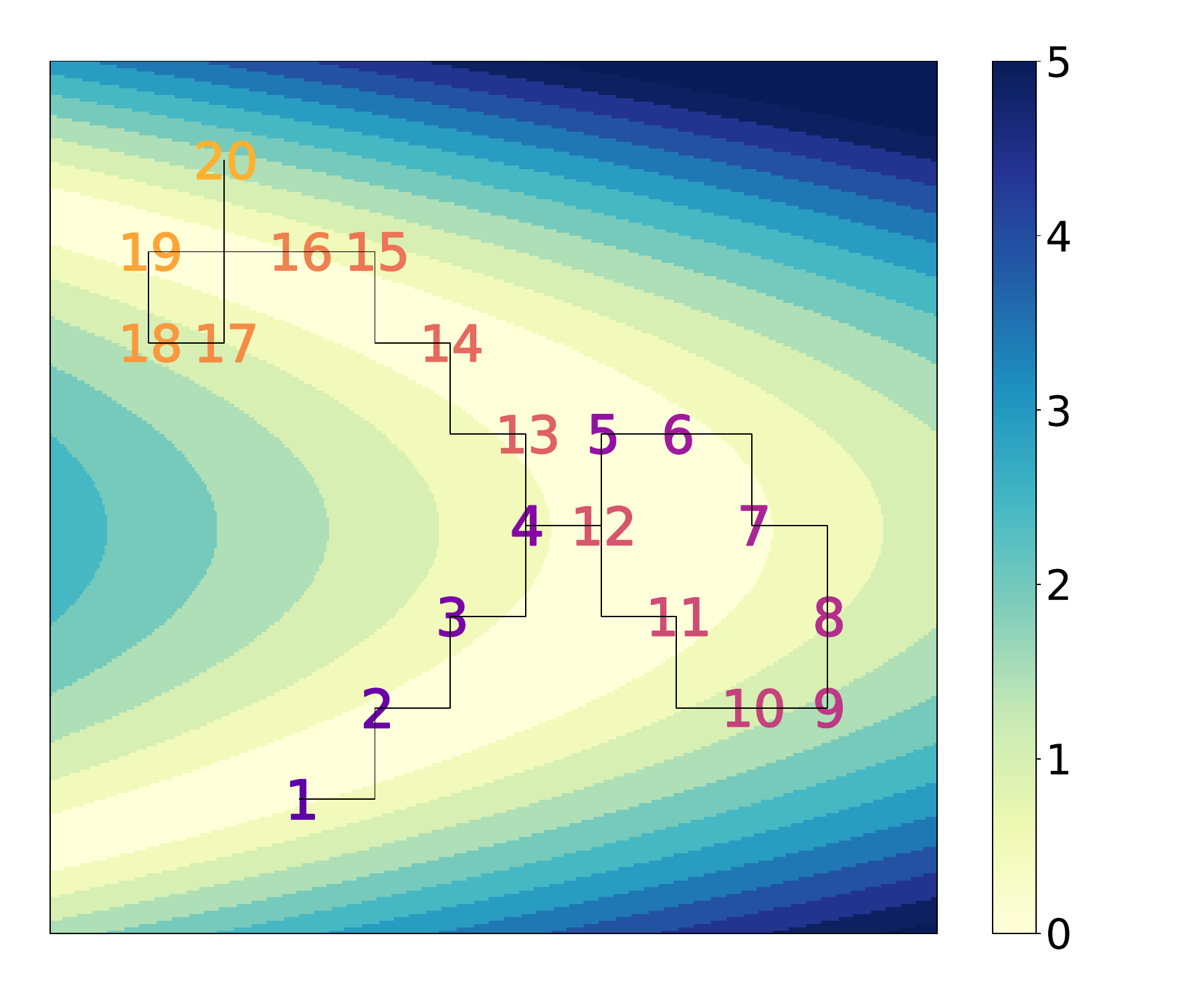}
	\end{minipage}
	\hspace{-.275\linewidth}
	\begin{minipage}{.45\linewidth}
	\caption{A trajectory of Laplace momentum based Hamiltonian dynamics $\{\theta_1(t), \theta_2(t)\}$ approximated by the coordinate-wise integrator of Algorithm~\ref{alg:coord_integrator_with_laplace_mom}. The log density of the target distribution changes in the increment of $0.5$ and has ``banana-shaped'' contours. Each numerical integration step consists of the coordinate-wise update along the horizontal axis followed by that along the vertical axis. The order of the coordinate updates is randomized at the beginning of numerical integration to ensure reversibility. The trajectory initially travels in the direction of the initial momentum, a process marked by the numbers \mbox{1 -- 5}. At the 5th numerical integration step, however, the trajectory does not have sufficient kinetic energy to take a step upward and hence flips the momentum downward. Such momentum flips also occur at the 8th and 9th numerical integration steps, again changing the direction of the trajectory. The rest of the trajectory follows the same pattern.}
	\label{fig:2d_traj_of_coord_integrator}
	\end{minipage}
	\vspace{-.5\baselineskip}
\end{figure}

\noindent More precisely, we assume that the parameter space has a partition $\mathbb{R}^{|I|} \times \mathbb{R}^{|J|} = \cup_k \mathbb{R}^{|I|} \times \Omega_k$ such that $U(\btheta)$ is smooth on $\mathbb{R}^{|I|} \times \Omega_k$ for each $k$. We write $\p = (\p_I, \p_J)$ correspondingly and define the distribution of $\p$ as a product of Gaussian and Laplace so that
\begin{equation}
\label{eq:mixed_kinetic_energy}
K(\p) 
	= - \log \pi_P(\p)
	= \frac{1}{2} \, \p_I^\intercal \mass_I^{-1} \p_I + \textstyle{\sum}_{j \in J} m_j^{-1} |p_j |,
\end{equation}
where $\mass_I$ and $\mass_J = \textrm{diag}(\bm{m}_J)$ are \textit{mass matrices} \citep{neal10}. With slight abuse of terminology, we will refer to $(\btheta_J, \p_J)$ as discontinuous parameters.

When mixing Gaussian and Laplace momenta, we approximate the dynamics via an integrator based again on operator splitting; we update the smooth parameter $(\btheta_I, \p_I)$ first, then the discontinuous parameter $(\btheta_J, \p_J)$, followed by another update of $(\btheta_I, \p_I)$. The discontinuous parameters are updated coordinate-wise as described in Section~\ref{sec:integrator_for_laplace_mom}. The update of $(\btheta_I, \p_I)$ is based on a decomposition familiar from the leapfrog scheme:
\begin{alignat}{3}
\frac{\diff \p_I}{\diff t} 
&= \nabla_{\btheta_I} \log \pi(\btheta), \quad &
\frac{\diff \btheta_I}{\diff t} 
&= \bm{0}, \quad &
\frac{\diff \btheta_J}{\diff t} &= \frac{\diff \p_J}{\diff t} = \bm{0}, 
\label{eq:smooth_momentum_update}\\
\frac{\diff \btheta_I}{\diff t} 
&= \mass_I^{-1} \p_I,  \ &
\frac{\diff \p_I}{\diff t} 
&= \bm{0}, \ &
\frac{\diff \btheta_J}{\diff t} &= \frac{\diff \p_J}{\diff t} = \bm{0}. \label{eq:smooth_position_update}
\end{alignat}
Algorithm~\ref{alg:dhmc_integrator} describes the integrator with all the ingredients put together.  When mixing in Gaussian momentum, the integrator continues to preserve the Hamiltonian accurately if not exactly, with the global error rate of $O(\epsilon^2)$ (supplement Section~\ref{sec:dhmc_integrator_error_analysis}). Advantages of separately treating continuous and discontinuous parameters in this manner are discussed in the supplement Section~\ref{sec:joint_vs_coordinate_update}.

\section{Theoretical properties of discontinuous \hmc{}}
\label{sec:theory_of_dhmc}

\subsection{Reversibility of \dhmc{} transition kernel}
\label{sec:dhmc_reversibility}

As for existing \hmc{} variants, the reversibility of \dhmc{} is a direct consequence of the reversibility and volume-preserving property of our integrator in Algorithm~\ref{alg:dhmc_integrator} \citep{neal10, fang14}. We thus focus on establishing these properties of our integrator.
Let $\bPsi$ denote a bijective map on the space $(\btheta, \p)$ corresponding to the approximation of discontinuous Hamiltonian dynamics through multiple numerical integration steps. An integrator is \textit{reversible} if $\bPsi$ satisfies
\begin{equation}
\label{eq:reversibility}
(\R \circ \bPsi)^{-1} = \R \circ \bPsi \quad
\text{ or equivalently } \quad \bPsi^{-1} = \R \circ \bPsi \circ \R,
\end{equation}
where $\R: (\btheta, \p) \to (\btheta, - \p)$ is the momentum flip operator. Due to the updates of discrete parameters in a random order, the map $\bPsi$ induced by our integrator is non-deterministic. Consequently, our integrator has an unconventional feature of being reversible ``in distribution'' only, $(\R \circ \bPsi)^{-1} \overset{d}{=} \R \circ \bPsi$, which is sufficient for the resulting Markov chain to be reversible.

\begin{lemma}
	\label{lem:symplecticity_coord_integrator}
	For a piecewise smooth potential energy $\pot(\btheta)$, the coordinate-wise integrator of Algorithm~\ref{alg:coord_integrator_with_laplace_mom} is volume-preserving and reversible for any coordinate index $i$ except on a set of Lebesgue measure zero. 
	Moreover, updating multiple coordinates with the integrator in a random order $\varphi(1), \ldots, \varphi(d)$ is again reversible in distribution provided that the random permutation $\varphi$ satisfies $\left\{\varphi(1), \varphi(2), \ldots, \varphi(d)\right\} \overset{d}{=} \{\varphi(d), \varphi(d - 1), \ldots, \varphi(1)\}$.
\end{lemma}
\begin{theorem}
	\label{thm:dhmc_integrator}
	For a piecewise smooth potential energy $\pot(\btheta)$, the integrator of Algorithm~\ref{alg:dhmc_integrator} is  volume-preserving and reversible except on a set of Lebesgue measure zero.
\end{theorem}
\noindent The proofs are in the supplement Section~\ref{sec:dhmc_proofs_and_additional_results}. 
In the same section, we also establish the reversibility and volume-preserving property of discontinuous Hamiltonian dynamics under alternative kinetic energies.

\subsection{Irreducibility via randomized stepsize}
\label{sec:ergodicity_of_dhmc}
Reducible behaviors in \hmc{} are rarely observed in practice despite the subtleties in theoretical analysis \citep{livingstone16, bou2017randomized-hmc, durmus2017hmc-convergence}.
However, care needs to be taken when applying the coordinate-wise integrator for \dhmc{}; its use with a fixed stepsize $\epsilon$ results in a reducible Markov chain which is not ergodic.
To see the issue, consider the transition probability of multiple iterations of \dhmc{} based on the integrator of Algorithm~\ref{alg:dhmc_integrator}. Given the initial state $\btheta_0$, the integrator of Algorithm~\ref{alg:coord_integrator_with_laplace_mom} moves the $i$-th coordinate of $\btheta$ only by the distance $\pm \epsilon m_i^{-1}$ regardless of the values of the momentum variable. The transition probability in the $\btheta$-space with $\p$ marginalized out, therefore, is supported on a grid
\begin{equation}
\Omega = \left\{ (\btheta_I, \btheta_J) : \btheta_J = \btheta_{0, J} + \epsilon \bm{m} \odot \bm{k} \text{ for a vector of integers $\bm{k}$} \right\},
\end{equation}
where $\btheta_J$ as in \eqref{eq:disc_indiices}  denotes the coordinates of $\btheta$ with discontinuous conditionals. 

This pathological behavior can be avoided by randomizing the stepsize at each iteration, say $\epsilon \sim \textrm{Unif}(\epsilon_{\min}, \epsilon_{\max})$. Randomizing the stepsize additionally addresses a possibility that smaller stepsizes are required in some regions of the parameter space \citep{neal10}. While the coordinate-wise integrator does not suffer from the stability issue of the leapfrog scheme, the quantity $\epsilon m_i^{-1}$ nonetheless needs to match the length scale of $\theta_i$; see the supplement Section~\ref{sec:tuning_dhmc}.

\subsection{Metropolis-within-Gibbs with momentum as special case}
\label{sec:relation_to_metropolis_gibbs}
Consider a version of \dhmc{} in which all the parameters are updated with the coordinate-wise integrator of Algorithm~\ref{alg:coord_integrator_with_laplace_mom}; in other words, the integrator of Algorithm~\ref{alg:dhmc_integrator} is applied with $J = \{1, \ldots, d\}$ and an empty indexing set $I$. This version is a generalization of random-scan Metropolis-within-Gibbs, also known as one-variable-at-a-time Metropolis. We therefore refer to this version as \textit{Metropolis-within-Gibbs with momentum}.

We use $\pi_{\mathcal{E}}(\cdot)$ and $\pi_{\Phi}(\cdot)$ to denote the distribution of a stepsize $\epsilon$ and of a permutation $\varphi$ of $\{1, \ldots, d\}$, where $\pi_{\Phi}(\cdot)$ satisfies $\left\{\varphi(1), \ldots, \varphi(d)\right\} \overset{d}{=} \{\varphi(d), \ldots, \varphi(1)\}$. With these notations, each iteration of Metropolis-within-Gibbs with momentum can be expressed as follows:
\begin{enumerate}
	\item Draw $\epsilon \sim \pi_{\mathcal{E}}(\cdot)$, $\varphi \sim \pi_{\Phi}(\cdot)$, and $p_j \sim \textrm{Laplace}(\textrm{scale} = m_j)$ for $j = 1, \ldots, d$.
	\item Repeat for $L$ times a sequential update of the coordinate $(\theta_j, p_j)$ for $j = \varphi(1), \ldots, \varphi(d)$ via  Algorithm~\ref{alg:coord_integrator_with_laplace_mom} with stepsize $\epsilon$.
\end{enumerate}
In this version of \dhmc{}, the integrator exactly preserves the Hamiltonian and the acceptance-rejection step can be omitted. 

When $L = 1$, the above algorithm recovers random-scan Metropolis-within-Gibbs. This can be seen by realizing that Lines~\ref{line:compare_energy} -- \ref{line:mom_flip} of Algorithm~\ref{alg:coord_integrator_with_laplace_mom} coincide with the standard Metropolis acceptance-rejection procedure for $\theta_j$. More precisely, the coordinate-wise integrator updates $\theta_j$ to $\theta_j + \epsilon m_j^{-1} \textrm{sign}(p_j)$ only if
\vspace{-.25\baselineskip}
\begin{equation}
\exp \!\left\{ - \pot(\btheta^*) + \pot(\btheta) \right\} > \exp \big( - m_j^{-1} |p_j| \big) \overset{d}{=} \textrm{Unif}(0, 1),
\end{equation}
where the last distributional equality follows from the fact $m_j^{-1} |p_j| \overset{d}{=} \textrm{Exp}(1)$. To summarize, when taking only one numerical integration step, the version of \dhmc{} considered here coincides with Metropolis-within-Gibbs with a random scan order $\varphi \sim \pi_\Phi(\cdot)$ and a symmetric proposal $\theta_j \pm \epsilon m_j^{-1}$ for each parameter with $\epsilon \sim \pi_{\mathcal{E}}(\cdot)$.
We could also consider a version of \dhmc{} with a fixed stepsize $\epsilon = 1$ but with a mass matrix randomized $(m_1^{-1}, \ldots, m_d^{-1}) \sim \pi_{M^{-1}}(\cdot)$ before each numerical integration step; this version would correspond to a more standard Metropolis-within-Gibbs with independent univariate proposals. 

Being a generalization of Metropolis-within-Gibbs, \dhmc{} is guaranteed a superior performance:
\begin{corollary}
	Under any efficiency metric, which may account for computational costs per iteration, an optimally tuned \dhmc{} is guaranteed to outperform a class of random-scan Metropolis-within-Gibbs samplers as described above.
\end{corollary}
\noindent In particular, an optimally tuned \dhmc{} will inherit the geometric ergodicity of a corresponding Metropolis-within-Gibbs sampler, sufficient conditions for which are investigated in \cite{johnson13}.
In practice, the addition of momentum to Metropolis-within-Gibbs allows for a more efficient update of correlated parameters as empirically confirmed in the supplement Section~\ref{sec:dhmc_vs_gibbs}.

Besides being a generalization Metropolis-within-Gibbs, \dhmc{} has a curious connection to the zig-zag sampler, a state-of-the-art non-reversible Monte Carlo algorithm.
The Laplace-momentum based Hamiltonian dynamics exhibits behaviors remarkably similar to the piece-wise deterministic Markov process underlying the zig-zag sampler (supplement Section~\ref{sec:relation_to_zigzag_details}).

\section{Numerical results}
\label{sec:simulation}

\subsection{Experimental set-up, benchmarks, and efficiency metric}
We use two challenging posterior inference problems to demonstrate the efficiency of \dhmc{} as a general-purpose sampler.
Additional numerical results in the supplement Section~\ref{sec:additional_numerical_results} further illustrate the breadth of its capability. Codes to reproduce the simulation results are available at \texttt{https://github.com/aki-nishimura/discontinuous-hmc}.

Few general and efficient approaches currently exist for sampling from a discrete parameter or a discontinuous target density. For each problem, therefore, we pick a few most appropriate general-purpose samplers to benchmark against. \citet{chopin17} compare a variety of algorithms on posterior distributions of binary classification problems. One of their conclusions is that, while random-walk Metropolis is known to scale poorly in the number of parameters \citep{roberts97}, Metropolis with a properly adapted proposal covariance is competitive with alternatives even in a 180-dimensional space.
As one of our benchmarks, therefore, we use random-walk Metropolis with proposal covariances proportional to estimated target covariances \citep{roberts97, haario01}. 
When the conditional densities can be evaluated efficiently and no strong dependence exists among the parameters,  Metropolis-within-Gibbs with component-wise adaptation can scale better than joint sampling via random-walk Metropolis \citep{haario2005componentwiseadap}. This approach thus is used as another benchmark. 

For models with discrete parameters, we also compare with the \nutsGibbs{} approach \citep{pymc16}. Conditionally on discrete parameters, continuous parameters are updated by the \nuts{} of \cite{hoffman14}. The standard implementation then updates discrete parameters with univariate Metropolis, but here we implement full conditional univariate updates via manually-optimized multinomial samplings. In our examples, these multinomial samplings take little time relative to continuous parameter updates, tilting the comparison in favor of \nutsGibbs{}.
We use the identity mass matrix for the \nuts{} to make a fair comparison to \dhmc{} with the identity mass.

In all our numerical results, continuous parameters with range constraints are transformed into unconstrained ones to facilitate sampling. More precisely, the constraint $\theta > 0$ is handled by a log transform $\theta \to \log \theta$ and  $\theta \in [0, 1]$ by a logit transform $\theta \to \log \left\{ \theta / (1 - \theta) \right\}$ as done in Stan and PyMC  \citep{stan16,pymc16}. 
For each example, the stepsize and path length for \dhmc{} were manually adjusted over short preliminary runs by visually examining trace plots. The stepsize for the continuous parameter updates of \nutsGibbs{} was adjusted analogously.

Efficiencies of the algorithms are compared through effective sample sizes \citep{geyer11}. As is commonly done in the \mcmc{} literature, we compute the effective sample sizes of the first and second moment estimators for each parameter and report the minimum value across all the parameters. \expandafter\MakeUppercase \ess{}s are estimated using the method of batch means with 25 batches \citep{geyer11}, averaged over the estimates from 8 independent chains.

\subsection{Jolly-Seber model: estimation of unknown open population size and survival rate from multiple capture-recapture data}
\label{sec:jolly_seber_example}
The Jolly-Seber model and its extensions are widely used in ecology to estimate unknown population sizes along with related parameters of interest \citep{schwarz99}. The model is motivated by the following experimental design. Individuals from a particular species are captured, marked, and released back to the environment. This procedure is repeated over multiple capture occasions. At each occasion, the number of marked and unmarked individuals among the captured ones are recorded. Individuals survive from one capture occasion to another with an unknown survival rate. The population is assumed to be ``open'' so that individuals may enter, either through birth or immigration, or leave the area under study.

In order to be consistent with the literature on capture-recapture models, the notations within this section will deviate from the rest of the paper. Assuming that data are collected over $i = 1, \ldots, T$ capture occasions, the unknown parameters are $\{U_i, p_i\}_{i=1}^{T}$ and $\{\phi_i\}_{i=1}^{T-1}$, representing
\vspace{-.3\baselineskip}
\begin{equation*}
\begin{aligned}
U_i &= \, \text{number of unmarked animals right before the $i$th capture occasion;} \\
p_i &= \, \text{capture probability of each animal at the $i$th capture occasion;} \\
\phi_i &= \, \text{survival probability of each animal from the $i$th to $(i+1)$th capture occasion.}
\end{aligned}
\end{equation*}
\vspace{-.75\baselineskip}

\noindent We assign standard objective priors $p_i, \phi_i \sim \textrm{Unif}(0, 1)$ and $\pi(U_1) \propto U_1^{-1}$. The parameters $U_2, \ldots, U_T$ require a more complex prior elicitation; this is described in the supplement Section~\ref{sec:jolly_seber_details} along with the likelihood function and other details on the Jolly-Seber model.

We take the black-kneed capsid population data from \citet{jolly65} as summarized in \citet{seber82}. The data record the capture-recapture information over $T = 13$ successive capture occasions, giving rise to a 38-dimensional posterior distribution involving 13 discrete parameters. We use the log-transformed embedding for the discrete parameter $U_i$'s (Section~\ref{sec:embedding}).
The proposal covariance for random-walk Metropolis is chosen by pre-computing the true posterior covariance with a long adaptive Metropolis chain \citep{haario01} and then scaling it according to \citet{roberts97}.
\expandafter\MakeUppercase \dhmc{} can also take advantage of the posterior covariance information, so we also try using a diagonal mass matrix whose entries are set according to the estimated posterior variance of each parameter (supplement Section~\ref{sec:tuning_dhmc}).  Starting from stationarity, we run $10^4$ iterations of \dhmc{} and \nutsGibbs{} and $5 \times 10^5$ iterations of Metropolis.

The performance of each algorithm is summarized in Table~\ref{tab:js_performance_summary} where ``DHMC (diagonal)'' and ``DHMC (identity)'' indicate \dhmc{} with a diagonal and identity mass matrix respectively. The table clearly indicates a superior performance over \nutsGibbs{} and Metropolis with approximately 60 and 7-fold efficiency increase respectively when using a diagonal mass matrix. The posterior distribution exhibits high negative correlations between $U_i$ and $p_i$ (Figure~\ref{fig:js_2d_posterior}). All the algorithms record the worst \ess{} in $p_1$, but the mixing of \nutsGibbs{} suffers most as $U_i$ and $p_i$ are updated conditionally.

\begin{table}
	\caption{Performance summary of each algorithm on the Jolly-Serber model example. ``DHMC'' and ``ESS'' in the table stand for \dhmc{} and effective sample size. The term $(\pm \ldots)$ indicates the error estimate, twice the standard deviations, of our \ess{} estimators. Path length is averaged over each iteration. ``Iter time''  shows the computational time per iteration of each algorithm relative to the fastest one.}
	\label{tab:js_performance_summary} 
	\centering
	\begin{tabular}{c|c|c|c|c}
		& ESS per 100 samples & ESS per minute & Path length & Iter time \\ 
		\hline
		DHMC (diagonal) & 45.5 ($\pm$ 5.2) & 424 & 45 & 87.7 \\ 
		\hline
		DHMC (identity) & 24.1 ($\pm$ 2.6) & 126 & 77.5 & 157 \\ 
		\hline
		\nutsGibbs{} & 1.04 ($\pm$ 0.087) & 6.38 & 150 & 133 \\ 
		\hline
		Metropolis & 0.0714 ($\pm$ 0.016) & 58.5 & 1 & 1
	\end{tabular}
\end{table}

\subsection{Generalized Bayesian belief update based on loss functions}
\label{sec:pac_bayes_example}
Motivated by model misspecification and difficulty in modeling all aspects of a data generating process, \citet{bissiri16} propose a generalized Bayesian framework, which replaces the log-likelihood with a surrogate based on a utility function.  Given an additive loss $\ell(y, \btheta)$ for the data $y$ and parameter of interest $\btheta$, the prior $\pi(\btheta)$ is updated to obtain the generalized posterior: 
\begin{equation}
\label{eq:pac_posterior}
\pi_{\textrm{post}}(\btheta) \propto \exp\!\left\{- \ell(\by, \btheta)\right\} \pi(\btheta).
\end{equation}
While \eqref{eq:pac_posterior} coincides with a pseudo-likelihood type approach,  \citet{bissiri16} derives the formula as a coherent and optimal update from a decision theoretic perspective.

Here we consider a binary classification problem with an error-rate loss:
\begin{equation}
\label{eq:error_rate_loss}
\ell(\by, \bm{\beta}) = \textstyle{\sum}_{i = 1} \mathbbm{1} \left\{ y_i \bm{x}_i^\intercal \bm{\beta} < 0 \right\},
\end{equation}
where $y_i \in \{-1, 1\}$, $\bm{x}_i$ is a vector of predictors, and $\bm{\beta}$ is a regression coefficient. The target distribution of the form \eqref{eq:pac_posterior} based on the loss function \eqref{eq:error_rate_loss} is suggested as a challenging test case by \citet{chopin17}. 
We use the \textsc{SECOM} data from the \textsc{UCI} machine learning repository, which records various sensor data that can be used to predict the production quality of a semi-conductor, measured as ``pass'' or ``fail.'' We first remove the predictors with more than 20 missing cases and then remove the observations that still had missing predictors, leaving  1,477 cases with 376 predictors. All the predictors are normalized and the regression coefficients $\beta_i$'s are given $\normal(0, 1)$ priors. Figure~\ref{fig:pac_bayes_posterior_conditional} illustrates the complexity of the target distribution.

\begin{figure}
\vspace{-.25\baselineskip}
	\hspace{-.03\textwidth}
	\begin{minipage}[l]{.5\linewidth}
		\includegraphics[width=\linewidth]{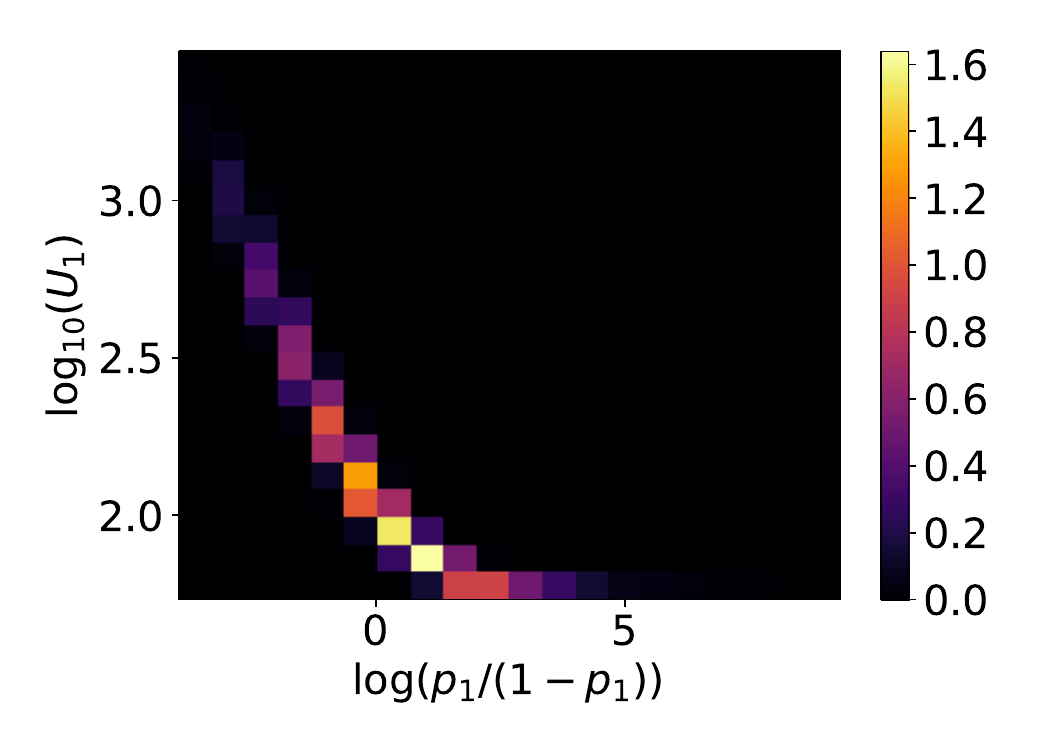}
	\end{minipage}
	~
	\begin{minipage}{.48\linewidth}
		\includegraphics[width=.95\linewidth]{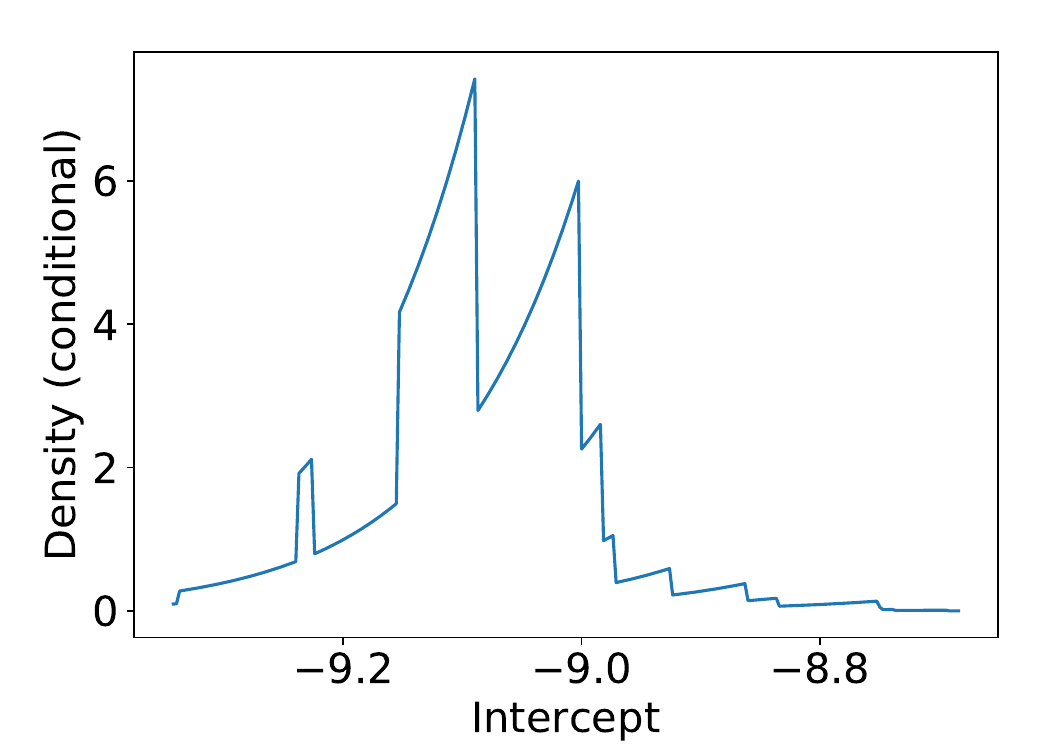}
	\end{minipage}
	\vspace{-2\baselineskip}
\end{figure}
\begin{figure}
	\hspace{-.25\textwidth}
	\begin{minipage}{.45\linewidth}
		\caption{The posterior marginal of $(p_1, U_1)$ with \newline parameter transformations, estimated from \newline the Monte Carlo samples.}
		\label{fig:js_2d_posterior}
	\end{minipage}
	\hspace{.03\textwidth}
	\begin{minipage}{.47\linewidth}
		\caption{The posterior conditional density of the intercept parameter in the generalized Bayes example. The other parameters are fixed at the posterior draw that recorded the highest posterior density among the Monte Carlo samples. The density is not continuous since the loss function is not.}
		\label{fig:pac_bayes_posterior_conditional} 
	\end{minipage}
\end{figure}

In tuning the proposal covariance of Metropolis for this example, adaptive Metropolis performed so poorly that we instead use $10^5$ iterations of \dhmc{} to estimate the posterior covariance. Scaling the proposal covariance for random-walk Metropolis according to \citet{roberts97} resulted in an acceptance probability of less than 0.04, so we scaled the proposal covariance to achieve the acceptance probability of 0.234 with stochastic optimization \citep{andrieu08}. We also found the posterior correlation to be very modest in this example, with the ratio of the largest to smallest eigenvalues of the estimated posterior covariance matrix being $46 \approx 6.8^2$. This suggested that coordinate-wise updates may be competitive, so we implemented Metropolis-within-Gibbs as an additional benchmark. The parameters are updated one at a time with the acceptance rate calibrated around 0.44 as recommended in \citet{gelman96}. 
We run \dhmc{} for $10^4$ iterations, Metropolis for $10^7$ iterations, and Metropolis-within-Gibbs for $5 \times 10^4$ iterations from stationarity.

Table~\ref{tab:pac_bayes_performance_summary} summarizes the performance of each algorithm. \expandafter\MakeUppercase \dhmc{} with identity mass matrix outperforms Metropolis and Metropolis-within-Gibbs by a factor of 330 and 2 respectively. Using a diagonal mass matrix yields only a minor improvement here as the posterior displays similar scales of uncertainty in all the parameters. The mixing of Metropolis suffers substantially from the dimensionality of the target. 
Conditional updates of Metropolis-within-Gibbs mix well in this example due to weak dependence among the parameters. On the other hand, as demonstrated in the example here and in Section~\ref{sec:jolly_seber_example}, \dhmc{} not only scales well in the number of parameters but also efficiently handles distributions with strong correlations. 

\begin{table}
	\caption{Performance summary of each algorithm on the generalized Bayesian posterior example. ``DHMC'' and ``ESS'' in the table stand for \dhmc{} and effective sample size. The term $(\pm \ldots)$ is the error estimate of our \ess{} estimators. Path length is averaged over each iteration. ``Iter time'' shows the computational time for one iteration of each algorithm relative to the fastest one.}
	\label{tab:pac_bayes_performance_summary} 
	\centering
	\begin{tabular}{c|c|c|c|c}
		& ESS per 100 samples & ESS per minute & Path length & Iter time \\ 
		\hline
		DHMC (identity) & 26.3 ($\pm$ 3.2) & 76 & 25 & 972 \\ 
		\hline
		Metropolis & 0.00809 ($\pm$ 0.0018) & 0.227 & 1 & 1 \\ 
		\hline
		Metropolis-within-Gibbs & 0.514 ($\pm$ 0.039) & 39.8 & 1 & 36.2
	\end{tabular}
\end{table}


\bibliographystyle{biometrika}
\bibliography{DHMC}{}%

\newcommand{\transpose}{\text{\raisebox{.5ex}{$\intercal$}}}
  
\newpage  
  
\markboth{Nishimura et~al.}{Discontinuous \hmc{}}  
  
\renewcommand{\thealgocf}{S\arabic{algocf}}  
\renewcommand{\thesection}{S\arabic{section}}   
\renewcommand{\thetable}{S\arabic{table}}     
\renewcommand{\thefigure}{S\arabic{figure}}  
\renewcommand{\theequation}{S\arabic{equation}}  
  
\setcounter{algocf}{0}  
\setcounter{section}{0}  
\setcounter{figure}{0}  
\setcounter{table}{0}  
\setcounter{equation}{0}  
  
{  
	\bigskip  
	\bigskip  
	\bigskip  
	\begin{center}  
		{\Large \bf  Supplement to ``Discontinuous Hamiltonian Monte Carlo for discrete parameters and discontinuous likelihoods'' \par}  
	\end{center}  
	\medskip  
}

\section{Behavior of leapfrog integrator on discontinuous target}
\label{sec:leapfrog_on_discontinuous_target}
As mentioned in Section~\ref{sec:hmc_fails_on_discontinuity}, in the presence of discontinuity, the leapfrog integrator in general incurs unbounded errors that do not decrease even as $\epsilon \to 0$. 
To see this, consider a discontinuous target $\pi_c(\theta)$ which is continuously differentiable except at $\theta = 0$, is constant on $\theta \in [-\delta, 0) \cup (0, \delta]$, and satisfies $\log \pi_c(-\delta) - \log \pi_c(\delta) = c > 0$. In particular, we have $\nabla \log \pi_c(\theta) = 0$ for $0 < |\theta| < \delta$ and hence the leapfrog trajectory evolves with a constant momentum on $\theta \in [-\delta, \delta]$. In other words,
\begin{equation*}
	\theta(n\epsilon) = \theta_0 + n \epsilon p_0, \ \ 
	p(n\epsilon) = p_0,
\end{equation*}
provided $|\theta_0 + k \epsilon p_0| < \delta$ for all $k = 0, 1, \ldots, n$. Starting from $\theta_0 < 0$,  when the leapfrog trajectory crosses $\theta = 0$ so that $\theta(k \epsilon) < 0 < \theta\{(k + 1) \epsilon\}$, it incurs the error of
\begin{equation*}
\begin{aligned}
&
	H\!\left[
		\theta\{(k + 1) \epsilon\}, p\{(k + 1) \epsilon\}
	\right]  - H\!\left\{ \theta(k \epsilon), p(k \epsilon) \right\} 
	\\
	&\hspace*{8em} = 
		- \log \pi_c[\theta\{(k + 1) \epsilon\}] + \log \pi_c\{\theta(k \epsilon)\}
	= c.
\end{aligned}
\end{equation*}

\section{Integrator for Gaussian momentum-based discontinuous dynamics}
\label{sec:ed_integrator_for_gaussian_mom}
Here we describe an implementation of the integrator proposed by \cite{pakman13} and \cite{afshar15}. The integrator is designed to approximate a discontinuous Hamiltonian dynamics with a Gaussian momentum corresponding to the kinetic energy $K(\p) = \| \p \|^2 / 2$. For simplicity, we assume that a parameter space $\Theta$ consists only of the embedded discrete parameters as described in Section~\ref{sec:embedding}, so that the target $\pi_{\Theta}(\cdot)$ is piecewise constant with the discontinuity set consisting of the boundaries of hyper-cubes.  The integrator is energy-preserving in this simplified setting but not so for more general discontinuous dynamics. The pseudo code is given in Algorithm~\ref{alg:ed_integrator_with_gaussian_mom}.


\begin{algorithm}
	\caption{Integrator for Gaussian momentum-based discontinuous dynamics}
	\label{alg:ed_integrator_with_gaussian_mom}
	\DontPrintSemicolon
	\SetKwInOut{Input}{Input}
	\Input{ initial state $(\btheta, \p)$, stepsize $\epsilon$ }
	\BlankLine
	$t \gets 0$ \;
	\While{$t < \epsilon$}{
		$t_e \gets$ the time until reaching the next discontinuity \;
		\eIf{$t + t_e > \epsilon$}{
			$\btheta \gets \btheta + (\epsilon - t) \p$ \;
			$t \gets \epsilon$ \;
		}{
			$\btheta \gets \btheta + t_e \p$ \;
			$i \gets$ the index of the axis orthogonal to the discontinuity plane at $\btheta$ \;
			$\Delta \pot_e \gets$ the potential energy difference  \;
			\eIf{$p_i^2 / 2 > \Delta \pot_e$}{
				$p_i \gets \sqrt{p_i^2 - 2 \Delta \pot_e}$ \;
			}{
				$p_i \gets - p_i$ \;
			}
			$t \gets t + t_e$
		}
	}
\end{algorithm}


\section{Empirical verification of ergodicity and unbiasedness of \dhmc{}}
\label{sec:dhmc_exactness_check}
To empirically back up the theoretical results of Section~\ref{sec:theory_of_dhmc}, here we use \dhmc{} to sample from a simple posterior distribution with closed-form marginal distributions. The correctness of \dhmc{} has been independently verified by \cite{gram2018dhmc-for-probabilistic-programs}, in which the \dhmc{} samples are compared to the outputs of existing probabilistic programming softwares.

We consider an observation model $y \given q, N \sim \mathrm{Binomial}(q, N)$ where both the success rate $q$ and sample size $N$ are unknown. We assign an objective prior $\pi(\nunknown) \propto \nunknown^{-1}$ \citep{berger12} and a beta prior $q \sim \textrm{Beta}(\alpha, \beta)$. As the particular choice is immaterial for the purpose of our simulation, we just pick $\alpha = \beta = 2$ and set $y = 100$. Closed-form expressions for the posterior marginals of $N$ and $q$ are given in Section~\ref{sec:dhmc_exactness_check_posterior_derivation} below.

To sample from the posterior, we use the log-transformed embedding of $N$ (Section~\ref{sec:embedding}). The parameter $q$ is mapped to the real line through a logit transform $q \to \log\{q / (1 - q)\}$. We use the integrator of Section~\ref{sec:mixing_momentum_dist} (Algorithm~\ref{alg:dhmc_integrator}) with the Laplace momentum for $N$ and Gaussian momentum for $q$. The stepsize $\epsilon$ is jittered in the range $[0.08, 0.1]$ and the number of numerical integration steps in the range $[15, 20]$. 

Figure~\ref{fig:dhmc_empirical_vs_exact_density} shows the empirical distributions of $N \given y$ and $q \given y$ from $10^6$ iterations of \dhmc{}. The empirical distributions are indistinguishable from the exact distributions indicated by the orange lines. Additionally, the trace plot in Figure~\ref{fig:dhmc_validity_check_traceplots} shows that \dhmc{} can induce a large transition in the parameter $N$ with only a small number of numerical integration steps. This means that the \dhmc{} integrator often jumps through a large number of discontinuities along the parameter $N$ at each numerical integration step. This behavior introduces no bias as the integrator remains reversible and volume-preserving regardless of its stepsize as discussed in the main manuscript Section~\ref{sec:theory_of_dhmc}.

\begin{figure}
	\includegraphics[width=\linewidth]{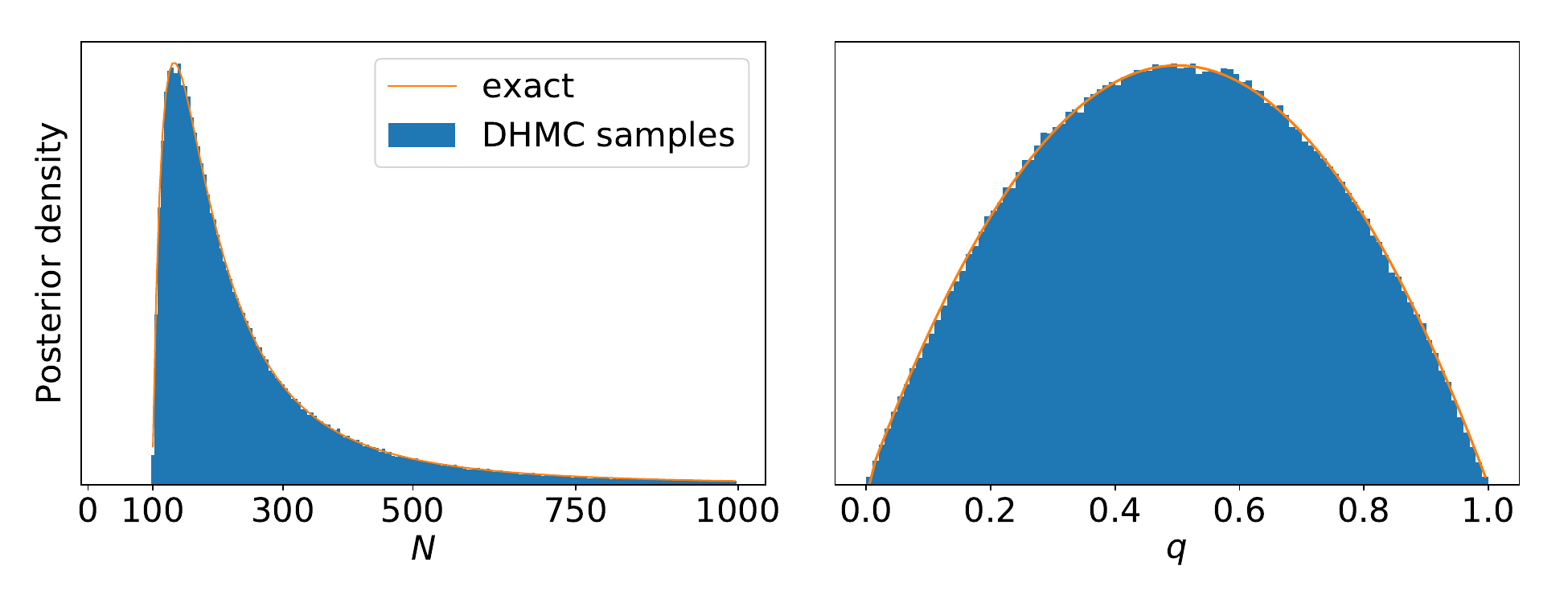}
	\caption{Empirical distributions of the \dhmc{} samples generated for the target distribution as described in Section~\ref{sec:dhmc_exactness_check_posterior_derivation}. The orange lines show the exact posterior mass and density functions computed from the closed-form expressions. The unknown sample size parameter $N$ has no posterior probability below the observed number of successes $y = 100$.}
	\label{fig:dhmc_empirical_vs_exact_density}
\end{figure}

\begin{figure}
	\includegraphics[width=\linewidth]{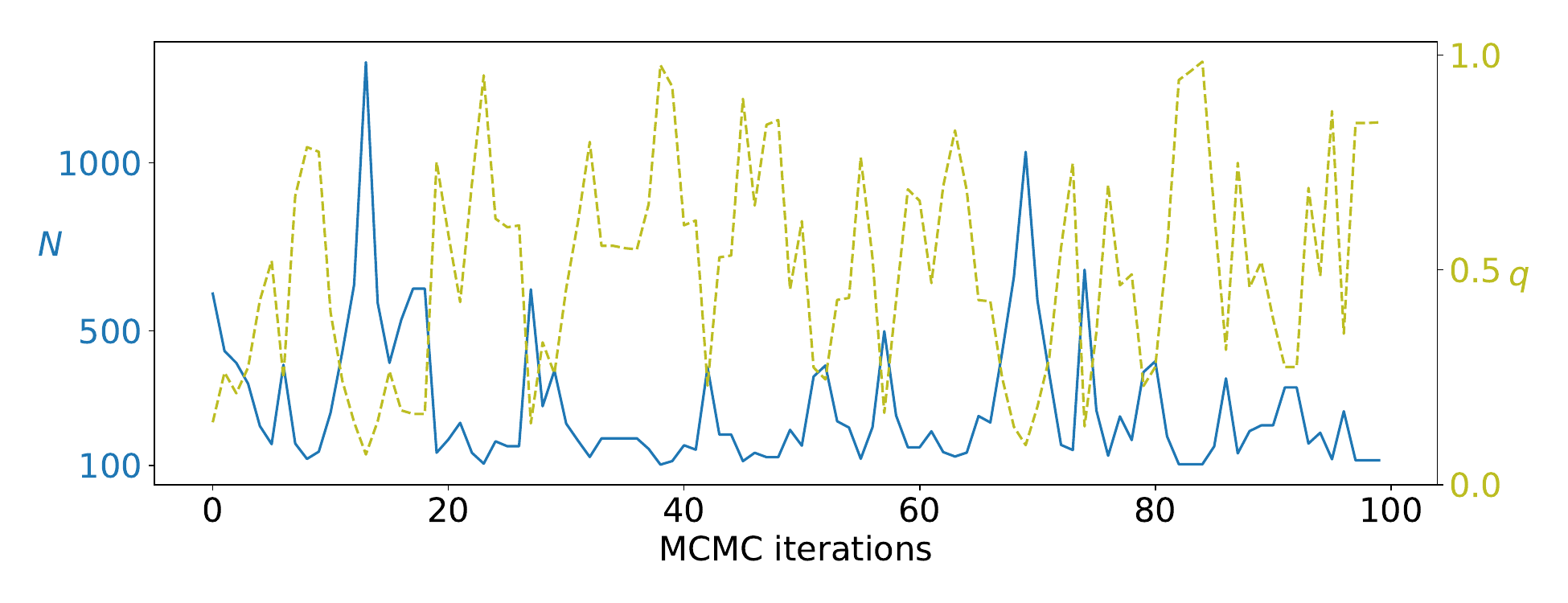}
	\caption{Trace plots for the first 100 \dhmc{} samples generated for the target distribution as described in Section~\ref{sec:dhmc_exactness_check_posterior_derivation}. The blue line and left $y$-axis indicates the parameter values of $N$, while the olive line and right $y$-axis indicates the parameter values of $q$.}
	\label{fig:dhmc_validity_check_traceplots}
\end{figure}

\subsection{Derivation of the posterior marginals}
\label{sec:dhmc_exactness_check_posterior_derivation}
For the model and priors described above, we have
\begin{equation}
\begin{aligned}
\pi(\nunknown, q \given \nobs)
	&\propto \frac{\nunknown !}{(\nunknown - \nobs)!} q^\nobs (1 - q)^{\nunknown - \nobs} \pi(q) \pi(N) 
	\propto \frac{(\nunknown - 1)!}{(\nunknown - \nobs)!} q^{\nobs + \alpha - 1} (1 - q)^{\nunknown - \nobs + \beta - 1} 
\end{aligned}
\end{equation}
Integrating over $q$, we obtain
\begin{equation}
\label{eq:toy_model_posterior_marginal_general}
\begin{aligned}
\pi(\nunknown \given \nobs)
	&\propto \frac{(\nunknown - 1)!}{(\nunknown - \nobs)!}
	\frac{
		\Gamma(\nunknown - \nobs + \beta) 
	}{
		\Gamma(\nunknown + \alpha + \beta)
	}
	=  \frac{(\nunknown - 1)!}{(\nunknown - \nobs)!}
	\frac{
		(\nunknown - \nobs + \beta - 1)! 
	}{
		(\nunknown + \alpha + \beta - 1)!
	} 
\end{aligned}
\end{equation}
where the equality holds when $\alpha$ and $\beta$ take positive integer values. 
The choice $\alpha = \beta = 2$ yields
\begin{equation}
\begin{aligned}
\pi(\nunknown \given \nobs)
	&\propto \frac{
		N - \nobs + 1
	}{
		(N + 3) (N + 2) (N + 1) N
	}
\end{aligned}
\end{equation}
We can compute the normalized mass function of $N \given y$ to high accuracy by truncating it at a suitably large number. Having computed $\pi(N \given y)$, we can compute the posterior marginal of $q$ via the law of total expectation $\pi(q \given y) = \sum_N \pi(q \given N, y) \pi(N \given y)$ where $q \given N, y \sim \mathrm{Beta}(y + \alpha, N - y + \beta)$.

\section{Relative advantages of joint and coordinate-wise updates on continuous parameters}
\label{sec:joint_vs_coordinate_update}
While the coordinate-wise update of Algorithm~\ref{alg:coord_integrator_with_laplace_mom} in the main manuscript generates a valid proposal whether or not $\pot(\btheta)$ has discontinuities along $\theta_i$, the joint update of continuous parameters as in Algorithm~\ref{alg:dhmc_integrator} has some computational advantages. First, when there is little conditional independence structure, calculating $\nabla_{\btheta_I} U(\btheta)$ is more computationally efficient than carrying out $|I|$ successive conditional density evaluations. Even when there is some conditional independence structure, however, computing $\nabla_{\btheta_I} U(\btheta)$ may still be substantially faster as an interpreter or compiler of a programing language can more easily optimize the required computation. Thus the joint update typically demands less computing time. On the other hand, the coordinate-wise updates have an advantage of being rejection-free by virtue of exact energy-preservation. The coordinate-wise update may thus be preferable for posteriors with substantial conditional independence structure such as those in latent Markov random field models.

\section{Tuning mass matrix and integrator stepsize of \dhmc{}}
\label{sec:tuning_dhmc}

\subsection{Role and tuning of mass matrix}
\label{sec:tuning_dhmc_mass}
As in the case of traditional \hmc{}, using a non-identity mass matrix has the effect of preconditioning a target distribution through reparametrization \citep{neal10}. More precisely, for a matrix $\bm{A}_I$ and a diagonal matrix $\bm{A}_J$, the performance of \dhmc{} under parametrization $(\bm{A}_I \btheta_I, \bm{A}_J \btheta_J, \p_I,  \p_J)$ is identical to that for $(\btheta_I, \btheta_J, \bm{A}_I^\intercal \p_I, \bm{A}_J^\intercal \p_J)$. The choice of a mass matrix for a Gaussian momentum is a well-studied topic \citep{neal10, girolami11}.
We can reason similarly for a Laplace momentum. We generally expect that sampling is facilitated by a reparametrization $\theta_j \to \theta_j / \textrm{var}(\theta_j)^{1/2} $ for $j \in J$. This is effectively achieved, given the relation between mass matrix choice and parameter transformation, by choosing the mass to be $m_j \approx \textrm{var}(\theta_j)^{-1/2}$. The variances can be estimated from a small number of preliminary \dhmc{} iterations. 

While the above discussion focuses on a diagonal mass matrix, it is also possible to encode the correlation structure of the target distribution into the Laplace momentum $p_J$. To precondition $\theta_J$ by reparametrization $\theta_J \to \mass_J^{-1/2} \theta_J$, we can define the distribution of $p_J$ to have independent Laplace distributions along the eigenvectors $u_1, \ldots, u_{k}$ of $\mass_J$:
	\begin{equation}
	p_J = \textstyle \sum_j \tilde{p}_j u_j 
		\ \text{ for } \, \tilde{p}_j \sim \text{Laplace}(\text{scale} = \delta_j),
	\end{equation}
so that $\text{Var}(p_J) = \mass_J = U D U^\transpose$ where $U = [u_1 | \ldots | u_k]$ and $D = \text{diag}(\delta_1, \ldots, \delta_k)$. The coordinate-wise integrator of Algorithm~\ref{alg:coord_integrator_with_laplace_mom} can then be applied along each $u_j$ one at a time. The new coordinate is likely to break the original conditional independence structures among the parameters, however, making each coordinate-wise update more expensive.
To incorporate a correlation structure while preserving some of the conditional independence structure, one possibility is to choose a block-diagonal $\mass_J$.
 
\subsection{Choice and tuning of integrator stepsize}
\label{sec:tuning_dhmc_stepsize}
The stepsize $\epsilon$ should be adjusted so that $\epsilon m_j^{-1}$ has the same order of magnitude as a typical scale of the conditional distribution of $\theta_j$. Unlike a leapfrog integrator that becomes unstable as $\epsilon$ increases, the coordinate-wise integrator remains exactly energy-preserving but at some point a large stepsize will cause \dhmc{} to ``get stuck'' at the current state. The numerical integration scheme of \dhmc{} will keep flipping the momentum $p_j \gets - p_j$ (Line~\ref{line:mom_flip} of Algorithm~\ref{alg:coord_integrator_with_laplace_mom}) without updating $\theta_j$  until the following condition is met:
\begin{equation}
\label{eq:coord_update_requirement}
\pot\big\{\btheta + \epsilon m_j^{-1} \textrm{sign}(p_j) \e_j\big\} - \pot(\btheta) < m_j^{-1} |p_j| \overset{d}{=} \textrm{Exponential}(1),
\end{equation}
where $\e_j$ denotes the $j$-th standard basis vector. When $\epsilon m_j^{-1}$ becomes larger than a typical scale of $\theta_j$, the condition \eqref{eq:coord_update_requirement} becomes unlikely to be satisfied, leading to infrequent updates of $\theta_j$.

We now consider how to tune the stepsize while the mass matrix fixed. This can be alternated with tuning of the mass matrix as suggested above to calibrate both the tuning parameters. To this end, we propose the following statistics:
\begin{equation}
\label{eq:stepsize_tune_statistic}
\begin{aligned}
&\mathbb{P}_{\pi_{\Theta} \times \pi_P} \left[ \pot\big\{\btheta + \epsilon m_j^{-1} \textrm{sign}(p_j) \e_j\big\} - \pot(\btheta) > m_j^{-1} |p_j|  \right] \\
&\hspace{3em} = \mathbb{E}_{\pi_{\Theta} \times \pi_P} \left\{ \min \left(1, \, \exp\Big[\pot(\btheta) -\pot\big\{\btheta + \epsilon m_j^{-1} \textrm{sign}(p_j) \e_j\big\}  \Big] \right) \right\}.
\end{aligned}
\end{equation}
The above statistics play a role analogous to the acceptance rate of Metropolis proposals. The statistics \eqref{eq:stepsize_tune_statistic} can be estimated, for example, by counting the frequency of momentum flips during each \dhmc{} iteration, and can then be used to tune the stepsize through stochastic optimization \citep{andrieu08, hoffman14}. One would want the statistics to be well above zero but not too close to 1, balancing the mixing rate and computational cost of each \dhmc{} iteration. Theoretical analysis of the optimal statistics value is beyond the scope of this paper, but the value 0.7 $\sim$ 0.9 is perhaps reasonable in analogy with the optimal acceptance rate for \hmc{} \citep{betancourt14}.

\section{Theoretical properties of discontinuous \hmc{}: proofs and additional results }
\label{sec:dhmc_proofs_and_additional_results}

\subsection{Proof of Lemma~\ref{lem:symplecticity_coord_integrator} and Theorem~\ref{thm:dhmc_integrator}}
\label{sec:reversibility_vol_preservation_proof}

\begin{proof}[Lemma~\ref{lem:symplecticity_coord_integrator}]
	Assume $p_i \neq 0$ for now and let $\e_i$ denote the $i$th standard basis vector. Then one step of Algorithm~\ref{alg:coord_integrator_with_laplace_mom} corresponds to a map $\bPsi_{i, \epsilon} : (\btheta, \p) \to (\btheta^*, \p^*)$ where
	\begin{equation}
	\label{eq:coord_update_pass}
	\btheta^* = \btheta + \epsilon m_i^{-1} \sign(p_i) \e_i, \quad \p^* = \p - m_i \left\{ \pot(\btheta^*) - \pot(\btheta) \right\} \e_i
	\end{equation}
	if $\pot\{\btheta + \epsilon m_i^{-1} \textrm{sign}(p_i) \e_i\} - \pot(\btheta) > m_i^{-1} p_i$, and otherwise
	\begin{equation}
	\label{eq:coord_update_bounce}
	\btheta^* = \btheta, \quad \p^* = -\p.
	\end{equation}
	The update equations \eqref{eq:coord_update_pass} and \eqref{eq:coord_update_bounce} are well-defined and differentiable except on the measure-zero set $S$, which we define momentarily. Under both \eqref{eq:coord_update_pass} and \eqref{eq:coord_update_bounce}, we have $\partial \btheta^* / \partial \p = 0$ and can easily show that
	\begin{equation}
	\det \left\{ \frac{\partial (\btheta^*, \p^*)}{\partial (\btheta, \p)} \right\}
	= \det \left( \frac{\partial \btheta^*}{\partial \btheta} \right) \det \left( \frac{\partial \p^*}{\partial \p} \right)
	= 1,
	\end{equation}
	establishing the volume-preservation. The reversibility as defined in \eqref{eq:reversibility} can be directly verified by solving the update equations \eqref{eq:coord_update_pass} and \eqref{eq:coord_update_bounce} for $(\btheta, -\p)$ as a function of $(\btheta^*, -\p^*)$. 
	
	We now quantify the set $S$ on which the above argument may break down and show that it has measure zero. Let $\mathcal{D}$ denote the discontinuity set of $\pot(\btheta)$ and $\mathcal{D} + \bv$ denote a set of points in $\mathcal{D}$ shifted by a vector $\bv$. It is easy to see that the update equations \eqref{eq:coord_update_pass} and \eqref{eq:coord_update_bounce} are well-defined and differentiable except when $(\btheta, \p)$ belongs to one of the sets below:
	\begin{equation}
	\label{eq:problematic_measure_zero_set}
	\mathcal{D} \times \mathbb{R}^d, \
	\left( \mathcal{D} \pm \epsilon m_i^{-1} \e_i \right) \times \mathbb{R}^d, \
	\left\{ p_i = 0 \right\}, \
	\left\{ \pot\{\btheta + \epsilon m_i^{-1} \textrm{sign}(p_i) \e_i\} - \pot(\btheta) = m_i^{-1} p_i \right\}.
	\end{equation}
	Each of these sets above corresponds to lower-dimensional manifolds of the parameter space and hence have measure zero. We define the set $S$ as the union of all the sets \eqref{eq:problematic_measure_zero_set} over $i = 1, \ldots, d$. Being a finite union of measure-zero sets, the set $S$ also has measure zero.
	
	Lastly, we prove the reversibility of multiple coordinate updates corresponding to a map $\bPsi_{\varphi(d), \epsilon} \circ \ldots \circ \bPsi_{\varphi(1), \epsilon}$ with a random permutation $\varphi$. From the reversibility of each $\bPsi_{i, \epsilon}$, we deduce that
	\begin{align}
	\R \circ \left( \bPsi_{\varphi(d), \epsilon} \circ \ldots \circ \bPsi_{\varphi(1), \epsilon} \right) \circ \R
	&= \bPsi_{\varphi(d), \epsilon}^{-1} \circ \ldots \circ \bPsi_{\varphi(1), \epsilon}^{-1}
	= \left(\bPsi_{\varphi(1), \epsilon} \circ \ldots \circ \bPsi_{\varphi(d), \epsilon} \right)^{-1}.
	\end{align}
	By our assumption on the distribution of $\varphi$, we have
	\begin{equation}
	\left(\bPsi_{\varphi(1), \epsilon} \circ \ldots \circ \bPsi_{\varphi(d), \epsilon} \right)^{-1}
	\overset{d}{=} \left(\bPsi_{\varphi(d), \epsilon} \circ \ldots \circ \bPsi_{\varphi(1), \epsilon} \right)^{-1}
	\end{equation}
	establishing the reversibility of $\bPsi_{\varphi(d), \epsilon} \circ \ldots \circ \bPsi_{\varphi(1), \epsilon}$ in distribution. 
\end{proof}

\begin{proof}[Theorem~\ref{thm:dhmc_integrator}]
	Let $\bPsi_{J, \, \varphi, \, \epsilon} = \bPsi_{\varphi(d'), \, \epsilon} \circ \ldots \circ  \bPsi_{\varphi(1), \, \epsilon}$ where $\bPsi_{j, \epsilon}: (\btheta, \p) \to (\btheta^*, \p^*)$ is defined as in \eqref{eq:coord_update_pass} and \eqref{eq:coord_update_bounce} and $\varphi(1), \ldots, \varphi(d')$ is a permutation of the indexing set $J$.  Also define $\bPsi_{\Theta, \, I, \,\epsilon / 2}$ and $\bPsi_{P, \, I, \, \epsilon / 2}$ as a function of $(\btheta, \p)$ such that
	\begin{equation}
	\bPsi_{\Theta, \, I, \,\epsilon / 2}: \btheta_I \to \btheta_I + \dfrac{\epsilon}{2} \, \mass_I^{-1} \p_I, \quad
	\bPsi_{P, \, I, \, \epsilon / 2}: \p_I \to \p_I - \dfrac{\epsilon}{2} \nabla_{\btheta_I} \pot(\btheta)
	\end{equation}
	while leaving all the other coordinates unchanged.
	The integrator of Algorithm~\ref{alg:dhmc_integrator} can then be formally expressed as a map
	\begin{equation}
	\label{eq:dhmc_integrator_operator}
	\bPsi_{\Theta, \, I, \,\epsilon / 2} \circ \bPsi_{P, \, I, \, \epsilon / 2} \circ \bPsi_{J, \, \varphi, \, \epsilon} \circ \bPsi_{P, \, I, \, \epsilon / 2} \circ \bPsi_{\Theta, \, I, \,\epsilon / 2}.
	\end{equation}
	Being a symmetric composition of reversible maps, the map \eqref{eq:dhmc_integrator_operator} is again reversible. The maps $\bPsi_{\Theta, \, I, \,\epsilon / 2} \circ \bPsi_{P, \, I, \, \epsilon / 2}$ and $\bPsi_{P, \, I, \, \epsilon / 2} \circ \bPsi_{\Theta, \, I, \,\epsilon / 2}$ coincide with symplectic Euler schemes in the coordinate $(\btheta_I, \p_I)$ and hence a\\
	re volume preserving \citep{hairer06}. Since $\bPsi_{J, \, \varphi, \, \epsilon}$ is also volume-preserving by the results of Lemma~\ref{lem:symplecticity_coord_integrator}, the composition \eqref{eq:dhmc_integrator_operator} is volume-preserving. 
\end{proof}

\subsection{Reversibility and volume-preserving property of discontinuous dynamics under alternative kinetic energies}
In Theorem~\ref{thm:symplecticity_disc_hamiltonian} below, we establish the reversibility and volume-preserving property of discontinuous Hamiltonian dynamics with alternative kinetic energies. Theorem~\ref{thm:symplecticity_disc_hamiltonian} extends the result of \citet{afshar15} and justifies the use of the Gaussian momentum-based integrator Algorithm~\ref{alg:ed_integrator_with_gaussian_mom} in the supplement.
A \textit{solution operator} $\bPsi_t$ of a differential equation, or more generally of a differential inclusion, is a map such that $\{\btheta(t), \p(t)\} = \bPsi_t(\btheta_0, \p_0)$ is a solution of the equation with the initial condition $\{\btheta(0), \p(0)\} = (\btheta_0, \p_0)$. Also, \textit{symplecticity} is a property of Hamiltonian dynamics which implies volume-preservation. Section~\ref{sec:symplecticity_proof} below provides the definition of symplecticity as well as the proof of Theorem~\ref{thm:symplecticity_disc_hamiltonian}.

\begin{theorem}
	\label{thm:symplecticity_disc_hamiltonian}
	Let $\pot(\btheta)$ be a piecewise constant potential energy function whose discontinuity set is piecewise linear. Suppose that a kinetic energy $K(\p)$ is symmetric, convex, piecewise smooth, and satisfies the growth condition $K(\p) \to \infty$ as $\Vert \p \Vert \to \infty$. Then the solution operator $\bPsi_t$ of discontinuous Hamiltonian dynamics as defined in
    Section~\ref{sec:event_driven_approach} is symplectic and reversible except on a set of Lebesgue measure zero.
\end{theorem}

\noindent Theorem~\ref{thm:symplecticity_disc_hamiltonian} generalizes the result of \citet{afshar15} to a larger class of kinetic energies, but we believe the conclusions can be extended to an even larger class of potential and kinetic energies. Such results may prove useful in constructing alternative approaches for dealing with discontinuous targets.

\subsection{Symplecticity of discontinuous Hamiltonian dynamics}
\label{sec:symplecticity_proof}

Here we establish the \textit{symplecticity} of discontinuous Hamiltonian dynamics under the assumptions of Theorem~\ref{thm:symplecticity_disc_hamiltonian}.
Symplecticity implies a volume preservation and further has important consequences in the stability of numerical approximation schemes \citep{hairer06}.

\begin{definition}
	\normalfont
	A differentiable map $(\btheta, \p) \to (\btheta^*, \p^*)$ is called \textit{symplectic} if
	\begin{equation}
	\label{eq:symplecticity}
	\frac{\partial (\btheta^*, \p^*)}{\partial (\btheta, \p)}^T \bm{J} \, \frac{\partial (\btheta^*, \p^*)}{\partial (\btheta, \p)} = \bm{J}
	\quad \text{ for } 
	\bm{J} = \begin{bmatrix}
	0 & \I_d \\
	- \I_d & 0
	\end{bmatrix},
	\end{equation}
	where $\I_d$ denotes a $d$-dimensional identity matrix. A dynamics is called symplectic if its solution operator is.
\end{definition}

\begin{proof}[of Theorem~\ref{thm:symplecticity_disc_hamiltonian}]
	Reversibility is a standard property of smooth Hamiltonian dynamics with a symmetric kinetic energy \citep{hairer06}. Defined as a point-wise limit of smooth dynamics, discontinuous dynamics therefore is also reversible.
	
	We turn to the proof of symplecticity. Under the assumption of Theorem~\ref{thm:symplecticity_disc_hamiltonian}, the evolution of discontinuous Hamiltonian dynamics from a state $(\btheta, \p)$ at $t = 0$ to $(\btheta^*, \p^*)$ at $t = \tau$ is given as follows. Dividing up the time intervals into a smaller pieces if necessary, we can without loss of generality assume that a trajectory $\left\{\btheta(t), \p(t) \right\}$ encounters only one discontinuity at $\btheta(t_e)$ during the interval $[0, \tau]$. Since $\pot(\btheta)$ is piecewise constant, the momentum remains constant and $\btheta(t)$ travels in a straight line except when hitting the discontinuity. The relationship between $(\btheta, \p)$ and $(\btheta^*, \p^*)$ is therefore given by
	\begin{equation}
	\label{eq:disc_dynamics_map}
	\begin{aligned}
	\btheta^* &= \btheta + t_e \nabla_{\p} K(\p) + (\tau - t_e) \nabla_{\p} K(\p^*) \\
	\p^* &= \p + \gamma(\p) \bnu_e
	\end{aligned}
	\end{equation}
	where $\gamma(\p)$ is defined implicitly as a solution of the following relations. If $\Delta \pot_e$ defined as in \eqref{eq:energy_diff} satisfies $\Delta \pot_e < K(\p) - \min_c K(\p - c \bnu_e)$, we define $\gamma(\p)$ as a solution of
	\begin{equation}
	\label{eq:implicit_gamma_definition_pass}
	K(\p - \gamma \bnu_e) = K(\p) + \Delta U_e  \quad \text{ with } \gamma > 0.
	\end{equation}
	Otherwise, $\gamma(\p)$ is defined through the relation:
	\begin{equation}
	\label{eq:implicit_gamma_definition_bounce}
	K(\p - \gamma \bnu_e) = K(\p) \quad \text{ with } \gamma > 0.
	\end{equation}
	The uniqueness of solutions to the above relations is guaranteed by the convexity and growth condition on $K(\p)$, and hence $\gamma(\p)$ is well-defined. The event time $t_e$ is also a function of $(\btheta, \p)$ and can easily be shown to be
	\begin{equation}
	\label{eq:event_time}
	t_e(\btheta, \p) = \frac{\alpha - \langle \btheta, \bnu_e \rangle}{\langle \nabla_{\p} K(\p), \bnu_e \rangle}, 
	\end{equation}
	where $\alpha$ is the distance from the origin of the discontinuity plane of $\pot$ at $\btheta(t_e)$. Assuming that $\btheta(t_e)$ is not at the intersection of the linear discontinuity planes and that $\Delta \pot_e \neq K(\p) - \min_c K(\p - c \bnu_e)$, the relation \eqref{eq:disc_dynamics_map} correctly describes the evolution of the dynamics on a neighborhood of $(\btheta, \p)$ with $\gamma(\p)$ defined either through \eqref{eq:implicit_gamma_definition_pass} or \eqref{eq:implicit_gamma_definition_bounce}. The map $(\btheta, \p) \to (\btheta^*, \p^*)$ therefore is differentiable and Lemma~\ref{lem:simplecticity_calculation} establishes the symplecticity through direct computation.
	
	Lastly, we turn to the almost everywhere differentiability of discontinuous Hamiltonian dynamics. To characterize where the solution operator fails to be differentiable, we first define the following sets:
	\begin{equation*}
	\begin{aligned}
	\mathcal{D} 
		&= \left\{\, \btheta : \text{ multiple discontinuity boundaries of $U$ intersects at $\btheta$} \right\}; \\
	\mathcal{U} 
		&= \left\{\, \Delta > 0 : \Delta = \pot(\btheta) - \pot(\btheta') \text{ for some } \btheta, \btheta' \right\}; \\
	\mathcal{V} 
		&= \left\{\, \bnu : \text{ $\bnu$ is orthonormal to a discontinuity boundary of $U$} \right\}.
	\end{aligned}
	\end{equation*}
	The above sets are all countable by our assumption on $\pot(\btheta)$. 
	Based on the behavior of a trajectory as described in the previous paragraph, a trajectory from the initial state $(\btheta_0, \p_0)$ potentially experiences a non-differentiable behavior at time $t$ only if the initial state belongs to one of the sets below:
	\begin{equation}
	\begin{aligned}
	\bigcup\limits_{\btheta \in \mathcal{D}} \left\{ (\btheta + s \nabla_{\p} K(\p), \p) : s \in \mathbb{R} \right\}, & \
	\bigcup\limits_{\Delta \in \mathcal{U}, \bnu \in \mathcal{V}}  \left\{ (\btheta, \p) : K(\p) - \min_c K(\p - c \bnu) = \Delta \right\} \\
	& \hspace{-5em} \left\{(\btheta, \p) : t = \frac{\alpha - \langle \btheta, \bnu_e \rangle}{\langle \nabla_{\p} K(\p), \bnu_e \rangle} \right\}.
	\end{aligned}
	\end{equation}
	Being a countable union of lower dimensional manifolds, the sets above all have measure zero. 
\end{proof}

\begin{lemma}
	\label{lem:simplecticity_calculation}
	The map \eqref{eq:disc_dynamics_map} is symplectic for $\gamma(\p)$ and $t_e(\btheta, \p)$ as defined through \eqref{eq:implicit_gamma_definition_pass}, \eqref{eq:implicit_gamma_definition_bounce}, and \eqref{eq:event_time}.
\end{lemma}

\begin{proof} 
	To simplify expressions, we denote $\bw = \nabla_{\p} K(\p)$, $\bw^* = \nabla_{\p} K(\p^*)$, and let $\hess$ and $\hess^*$ denote the Hessians of $K$ at $\p$ and $\p^*$. First, an implicit differentiation of either \eqref{eq:implicit_gamma_definition_pass} or \eqref{eq:implicit_gamma_definition_bounce} with some algebra yields
	\begin{equation}
	\frac{\partial \gamma}{\partial \p} = \frac{\bw^\intercal - \bw^{* \intercal}}{\langle \bw^*, \bnu \rangle}.
	\end{equation}
	Differentiating \eqref{eq:disc_dynamics_map} with respect to $(\btheta, \p)$, we obtain
	\begin{equation}
	\label{eq:disc_dynamics_deriv}
	\begin{aligned}
	\frac{\partial \btheta^*}{\partial \btheta} 
	&= \I - \frac{(\bw - \bw^*) \bnu_e^\intercal}{ \langle \bw, \bnu_e \rangle }, & \
	\frac{\partial \btheta^*}{\partial \p} 
	&= t_e \hess - \frac{t_e}{ \langle w, \bnu_e \rangle } (\bw - \bw^*) \bnu_e^\intercal  \hess + (\tau - t_e) \hess^* \frac{\partial \p^*}{\partial \p} \\
	\frac{\partial \p^*}{\partial \btheta}
	&= 0, & \   
	\frac{\partial \p^*}{\partial \p}  
	&= \I + \frac{\bnu_e (\bw - \bw^*)^\intercal}{\langle \bw^*, \bnu_e \rangle}.
	\end{aligned}
	\end{equation}
	When $\partial \p^* / \partial \btheta = \bm{0}$, the symplecticity condition \eqref{eq:symplecticity} simplifies to:
	\begin{equation}
	\label{eq:simplified_symplecticity}
	\frac{\partial \btheta^*}{\partial \btheta}^\intercal \frac{\partial \p^*}{\partial \p}  = \I, \quad
	\frac{\partial \p^*}{\partial \p}^\intercal \frac{\partial \btheta^*}{\partial \p} = \left( \frac{\partial \p^*}{\partial \p}^\intercal \frac{\partial \btheta^*}{\partial \p} \right)^\intercal.
	\end{equation}
	The first equality in \eqref{eq:simplified_symplecticity} is easily verified from \eqref{eq:disc_dynamics_deriv}. To establish the second equality of \eqref{eq:simplified_symplecticity}, we need to verify the symmetry of the matrix
	\begin{equation}
	\label{eq:symmetry_to_be_proven}
	\frac{\partial \p^*}{\partial \p}^\intercal \frac{\partial \btheta^*}{\partial \p} 
	= t_e \frac{\partial \p^*}{\partial \p}^\intercal \left\{ \I - \frac{(\bw - \bw^*) \bnu_e^\intercal}{ \langle \bw, \bnu_e \rangle } \right\} \hess + (\tau - t_e) \frac{\partial \p^*}{\partial \p}^\intercal  \hess^* \frac{\partial \p^*}{\partial \p}.
	\end{equation}
	The first term of \eqref{eq:symmetry_to_be_proven} simplifies to $t_e \hess$, which is symmetric, and the second term is obviously symmetric. 
\end{proof}

\subsection{Connections between zig-zag process and Laplace momentum-based Hamiltonian dynamics}
\label{sec:relation_to_zigzag_details}
The zig-zag sampler is a state-of-the-art non-reversible Monte Carlo algorithm based on a piece-wise deterministic Markov process called a \textit{zig-zag process} \citep{bierkens2016zigzag, fearnhead2016piecewisemarkov, bierkens2017piecewisemarkov}. 
Here we describe a remarkable similarity between a zig-zag process and the Laplace momentum-based Hamiltonian dynamics with unit mass $m_j = 1$. 

As described in Section~\ref{sec:dynamics_with_laplace_mom} of the main manuscript, this Hamiltonian dynamics is governed by the following differential equation:
\begin{equation}
\frac{\diff \btheta}{\diff t}
= \textrm{sign}(\p), \quad 
\frac{\diff \p}{\diff t}
= - \nabla_{\btheta} \pot(\btheta).
\end{equation}
Consider a zig-zag process and Hamiltonian dynamics both starting from the state $\btheta_0$. Let $\bm{v}_0$  drawn uniformly drawn from $\{-1, +1\}^d$ be the initial velocity of the zig-zag process and $\p_0 = (p_{0,1}, \ldots, p_{0,d})$ drawn from the independent Laplace distribution be the initial momentum of the Hamiltonian dynamics. Under both the zig-zag process and Hamiltonian dynamics, the velocities remain constant while the parameter $\btheta$ moves along a straight line $\btheta^Z(t) = \btheta_0 + t \bv_0$ and $\btheta^H(t) = \btheta_0 + t \, \textrm{sign}(\p_0) $ for $t > 0$ until their respective first event times. The first event time for the zig-zag process is given as $t_{e}^Z = \min\{t_1^Z, \ldots, t_d^Z \}$ where
\begin{equation}
\label{eq:zigzag_event_time}
t_i^Z = \inf_{t' > 0} \left\{ \tau_i = \int_{0}^{t'} \left[v_{0,i} \partial_{\theta_i} U (\btheta_0 + t \bv_0) \right]^+ \diff t' \right\} 
\end{equation}
with $[x]^+ = \max\{0, x\}$  and $\tau_i$'s drawn from $\mathrm{Exp}(1)$. For the Hamiltonian dynamics, the first event time is given as $t_e^H = \min\{t_1^H, \ldots, t_d^H\}$ where
\begin{equation}
\label{eq:hamilton_event_time}
t_i^H = \inf_{t' > 0} \left[ |p_{0,i}| = \int_{0}^{t'} \textrm{sign}(p_{0,i}) \, \partial_{\theta_i} U \!\left\{\btheta_0 + t \, \textrm{sign}(\p_0) \right\} \diff t' \right]
\end{equation}
For both processes, the events result in the velocity change $v_{k} \gets - v_k$ and $\textrm{sign}(p_\ell) \gets - \textrm{sign}(p_\ell)$ for $k = \textrm{argmin}_i \{t_i^Z\}$ and $\ell = \textrm{argmin}_i \{t_i^H\}$.

Given that $(\bv_0, \bm{\tau}) \overset{d}{=} \{\textrm{sign}(\p_0), |\p_0|\}$, the similarity between \eqref{eq:zigzag_event_time} and \eqref{eq:hamilton_event_time} is striking. In fact, if $U(\btheta)$ were convex and $\btheta_0$ was the minimum of $U(\btheta)$, then the two processes $\{\btheta^Z(t), 0 \leq t \leq t_e^Z\}$ and $\{\btheta^H(t), 0 \leq t \leq t_e^H \}$ coincide in distribution. After the first event time or in  more general settings, however, the two processes diverge because a zig-zag process $(\btheta^Z, \diff \btheta^Z \hspace{-.5ex}/ \diff t) = (\btheta^Z, \bv)$ is Markovian while its Hamiltonian dynamics counter-part $(\btheta^H, \diff \btheta^H \hspace{-.5ex}/ \diff t) = \{\btheta^H, \textrm{sign}(\p)\}$ is not. 
More precisely, Hamiltonian dynamics after each event retains the magnitudes of its momentum $|p_i|$'s from the previous moment, so that the future evolution of $\{\btheta^H, \textrm{sign}(\p)\}$ cannot be determined only from its current value without the magnitude information.
Also, Hamiltonian dynamics accumulates kinetic energy while potential energy goes downhill such that $\textrm{sign}\!\left\{p_{i}(t) \right\} \partial_{\theta_i} U\{\btheta^H(t)\} < 0$. This creates a tendency for each coordinate of a Hamiltonian dynamics trajectory $\btheta^H(t)$ to travel longer in the same direction before switching its direction compared to that of a zig-zag process.

Its close connection to a state-of-the-art sampler partially explains the empirical success of \dhmc{} in Section~\ref{sec:dhmc_vs_gibbs}, though the application of \dhmc{} to smooth target distributions is outside the main focus of this paper. Some potential advantages of the zig-zag sampler include its non-reversibility and the fact that its entire trajectory can be used as valid samples from the target. In fact, \dhmc{} can also be made non-reversible through partial momentum refreshments \citep{neal10} and can utilize the entire trajectories as valid samples \citep{nishimura15recycle}. These strategies will likely further boost the performance of \dhmc{}. 

\vspace{-.75\baselineskip}
\section{Additional details on Jolly-Seber model}
\label{sec:jolly_seber_details}

\newcommand{\tilB}{\tilde{B}}
\newcommand{\tilphi}{\tilde{\phi}}
\newcommand{\tilp}{\tilde{p}}

\subsection{Sufficient statistics and likelihood function}
Under appropriate assumptions, details of which we refer the reader to \cite{seber82}, the likelihood of the Jolly-Seber model depends only on the following statistics from a capture-recapture experiment carried over $i = 1, \ldots, T$ capture occasions:
\begin{equation*}
\begin{aligned}
R_i &= \, \text{number of marked animals released after the $i$th capture occasion;} \\
r_i &= \, \text{number of animals from the released $R_i$ animals that are subsequently captured;} \\
z_i &= \, \text{number of animals that are caught before $i$th capture occasion,} \\
	&\hspace*{2.5em} \text{not caught in the $i$th capture occasion, but caught subsequently;} \\
m_i &= \, \text{number of marked animals caught at the $i$th capture occasion;} \\
u_i &= \, \text{number of unmarked animals caught at the $i$th capture occasion.}
\end{aligned}
\end{equation*}
The likelihood decomposes into two parts: one for the first captures of previously unmarked animals and another for their re-captures.  More precisely,
\begin{equation}
\begin{aligned}
L(\text{data} \given \bm{U}, \bm{p}, \bm{\phi})
&= L(\text{first captures}) \times L(\text{re-captures})  \\
L(\text{first captures})
&\propto \prod_{i=1}^{T} \frac{U_i!}{U_i - u_i !} p_i^{u_i} (1 - p_i)^{U_i - u_i} \\
L(\text{re-captures}) 
&\propto \prod_{i=1}^{T-1} \chi_i^{R_i - r_i} \{\phi_i (1 - p_{i+1})\}^{z_{i+1}} (\phi_i p_{i+1})^{m_{i+1}} \\
\end{aligned}
\end{equation}
where $\chi_i$ represents the conditional probability that a marked animal released after the $i$th capture occasion is not caught again. Mathematically, $\chi_i$ is defined recursively as 
\begin{equation}
\begin{aligned}
\chi_{T-1} &= 1 - \phi_{T-1} p_T, \quad
\chi_i = 1 - \phi_i \{p_{i+1} + (1 - p_{i+1})(1 - \chi_{i+1})\}.
\end{aligned}
\end{equation}

\subsection{Prior distribution for $U_{i+1} \given U_i, \phi_i$}
\label{sec:js_unmarked_prior_conditionals}
Let $B_i$ denote the number of ``births,'' representing animals that are born, enter (immigration), or leave (emigration) the population after the $i$th occasion and remain so until the $(i+1)$th occasion.  Also let $S_i$ denote the number of animals that are unmarked right after the $i$th capture occasion and survive until the next capture occasion. Then we have $U_{i+1} = B_i + S_i$ where $S_i \given U_i, u_i, \phi_i
\sim \textrm{Binomial}(\phi_i, U_i - u_i)$.

The prior distribution of $\{U_i\}_{i=1}^{T}$ can thus be induced by assigning a prior on  $B_i$'s. In our example, we assign a convenient prior on $U_i$'s based on the assumptions that 1) $\textrm{Binomial}(\phi_i, U_i - u_i)$ can be approximated by $\normal\{\phi_i (U_i - u_i), \phi_i (1 - \phi_i)\}$ and 2) $B_i$'s approximately follows independent $\normal(0, \sigma_B^2)$ . These assumptions motivate a prior
\begin{equation}
U_{i+1} \given U_i, u_i, \phi_i, \sigma_B
\sim \left\lfloor \normal\{\phi_i (U_i - u_i), \sigma_B^2 + \phi_i (1 - \phi_i)\} \right\rfloor,
\end{equation}
where $\lfloor \cdot \rfloor$ is a floor function. We used $\sigma_B = 500$ in our example of Section~\ref{sec:jolly_seber_example} in the main manuscript. An alternative prior on $\{U_i\}_{i=1}^{T}$ can be assigned to reflect different model and prior assumptions on the number of births. For instance, it is more natural to constrain $B_i \geq 0$ in some cases \citep{schwarz96} and a binomial distribution on $B_i$ will for example induce a Poisson-binomial distribution on the conditional $U_{i + 1} \given U_i, u_i, \phi_i$ after marginalizing over $B_i$ and $S_i$.

\subsection{Inference on unknown population sizes}
\label{sec:js_function_of_unknown_params}
In case the total population sizes $\{N_i\}_{i=1}^T$ at each capture occasion are of interest, we can generate their posterior samples using the relation $N_i = M_i + U_i$ where $M_i$ denotes the number of marked animals right before the $(i+1)$th capture occasion. The distribution of $\{M_i\}_{i=1}^T$ follows $M_0 = 0$ and $M_{i+1} \given M_i, \phi_i \sim \mathrm{Binomial}(M_i, \phi_i)$.

\section{Additional numerical results}
\label{sec:additional_numerical_results}

\subsection{Comparison of \dhmc{} and Gibbs in synthetic example}
\label{sec:dhmc_vs_gibbs}
We use a synthetic target distribution to demonstrate the difference between \linebreak Metropolis-within-Gibbs with and without momentum as discussed in the main manuscript Section~\ref{sec:relation_to_metropolis_gibbs}. While \dhmc{} requires neither conjugacy or smoothness of the conditional densities, we choose a multivariate Gaussian target distribution so that we can compare \dhmc{} to an optimal Metropolis-within-Gibbs implementation with the univariate proposal variances chosen according to the theory of \cite{gelman96}. In particular, we assume that the target distribution of $\btheta$ follows that of a stationary unit variance auto-regressive process of the form
\begin{equation}
\theta_t = \alpha \theta_{t-1} + \sqrt{1 - \alpha^2} \eta_t, \quad \theta_1, \eta_t \sim \normal(0, 1)
\end{equation}
for $t = 2, \ldots, 1000$ with $\alpha = 0.9$.

We compare the performances of four algorithms: \dhmc{} (coordinate-wise), Gibbs (full conditional updates), Metropolis-within-Gibbs (univariate updates with optimal proposal variances), and the no-U-turn sampler of \citet{hoffman14}. The performance of each algorithm is summarized in Table~\ref{tab:ar_performance_summary}. 
Remarkably, \dhmc{} outperforms not only Metropolis-within-Gibbs but also Gibbs, despite requiring no closed-form conditionals at all. After accounting for the computational costs, \dhmc{} improves Gibbs by over 50\% and Metropolis-within-Gibbs by over 600\%. In general, the  advantage of \dhmc{} over Gibbs is expected to increase as the correlations among the parameters increase because the use of momentum can suppress the ``random walk behavior'' \citep{neal10}. The covariance matrix of the target distribution here has a condition number $\approx 19^2$, which corresponds to substantial but not particularly severe correlations.

In computing \ess{} per unit time, we estimated theoretical and platform-independent relative computational time of the algorithms as follows. 
In reasonable low-level language implementations, the computation of conditional densities should account for the majority of computational times for a typical target distribution. Therefore, computational efforts should be roughly equivalent between one numerical integration step of \dhmc{} and one iteration of the Metropolis-within-Gibbs sampler. The computational cost of the \nuts{} and Gibbs relative to these algorithms is more specific to individual target distributions, depending strongly on specific structures such as conditional independence among the parameters. For this reason, we do not attempt to compare the \nuts{} and Gibbs to the other algorithms in terms of \ess{} per unit time.

\begin{table}
	\caption{Performance summary of each algorithm on the auto-regressive process example. ``DHMC'' and ``ESS'' in the table stands for \dhmc{} and \ess{}. The term $(\pm \ldots)$ indicates the error estimate of our \ess{} estimators. 
	Path length is averaged over each iteration. ``Iter time'' shows the computational time for one iteration of each algorithm relative to the fastest one.}
	\label{tab:ar_performance_summary}
	\centering
	\resizebox{\columnwidth}{!}{
		\begin{tabular}{c|c|c|c|c}
			& ESS per 100 samples & ESS per unit time & Path length & Iter time \\ 
			\hline
			DHMC & 77.4 ($\pm$ 5.2) & 7.12 & 49.5 & 49.5 \\ 
			\hline
			No-U-turn & 52.4 ($\pm$ 3.2) & N/A & 142 & N/A \\ 
			\hline
			Gibbs & 0.949 ($\pm$ 0.076) & 4.33 & N/A & N/A \\ 
			\hline
			Metropolis-within-Gibbs & 0.219 ($\pm$ 0.015) & 1 & N/A & 1
		\end{tabular}
	}
\end{table}

\subsection{Multiple change-point detection for auto-regressive	conditional heteroscedastic processes}
\label{sec:change_point_detection_example}
Auto-regressive conditional heteroscedastic processes are popular models for log-returns of speculative prices such as stock market indices. A non-stationary first-order auto-regressive conditional heteroscedastic process $\{y_t\}_{t=1}^T$ with parameters $\{a(t), b(t)\}_{t=1}^T$ assumes the distribution
\begin{equation}
y_t \given y_{t-1}, a, b
\sim \normal(0, \sigma_t^2) \quad \text{ where } \ \sigma_t^2 = a(t) + b(t) \, y_{t-1}^2.
\end{equation}
Motivated by its interpretability and advantage in forecasting, \cite{fryzlewicz14} propose a piecewise constant parametrization of $a(t)$ and $b(t)$ as follows:
\begin{equation}
\label{eq:change_point_vol_parametrization}
\big\{a(t), b(t) \big\} = (a_k, b_k) \ \text{ if } \ \tau_{k - 1} < t \leq \tau_{k}
\end{equation}
for $k = 1, \ldots, K$, where the number of change points $K$ and their locations $1 = \tau_0 < \tau_1 < \ldots < \tau_K$ are to be estimated along with $(a_k, b_k)$'s.

To fit the above model within a Bayesian paradigm, we infer the change points  through a variable selection type approach as follows, using the horseshoe shrinkage priors of \cite{carvalho10}. We first choose an upper bound $K_{\max}$ on the number of change points and assume a uniform prior on $\tau_k$'s on the constrained space $1 < \tau_1 < \ldots < \tau_{K_{\max}} < T$.  We then model the changes in the values of $a(t)$ and $b(t)$ through a prior
\begin{equation}
\begin{array}{c}
\log(a_k / a_{k-1}) \sim \normal\left(0, \sigma_a \eta_{a,k}\right) \\
\log(b_k / b_{k-1}) \sim \normal\left(0, \sigma_b \eta_{b,k}\right)
\end{array} 
\quad \text{ with } \  \eta_{a, k}, \eta_{b, k} \sim \textrm{Cauchy}^+(0, 1),
\end{equation}
where $\textrm{Cauchy}^+(0, 1)$ denotes the standard half-Cauchy prior and $\sigma_a$ and $\sigma_b$ are the global shrinkage parameters \citep{carvalho10}. The above approach can  ``select'' a subset of $\tau_1, \ldots, \tau_{K_{\max}}$ as real change points by removing the others through shrinkage $a_k \approx a_{k-1}$ and $b_k \approx b_{k-1}$.  We place a default prior $\sigma_a, \sigma_b \sim \textrm{Cauchy}^+(0, 1)$ for the global shrinkage parameters \citep{gelman06}, and $a_0, b_0 \sim \normal(0, 1)$ for the initial volatility parameters.

Following \cite{fryzlewicz14}, we fit our model to the log-return values of a stock market index over a period that includes the subprime mortgage crisis. In particular, we use the daily closing values of S\&P 500 on the market opening days during the period from Jan 1st, 2005 to Dec 31st, 2009. The log-return value cannot be computed when a daily closing value exactly coincides with the previous one; there were four such days during the period and these data points were removed. The model parameters in this example are largely nonidentifiable even with the order constraint $\tau_1, \ldots, \tau_{K_{\max}}$. In such cases, it is not clear if the minimum \ess{} across the individual parameters is a good measure of efficiency. For this example, therefore, we calculate the minimum \ess{} over the first and second moments of the following quantities: the hyper-parameters $\sigma_a$ and $\sigma_b$, log posterior density, and four summary statistics of the estimated functions $a(t)$ and $b(t)$. The four summary statistics $\log(\| a \|_2), \log(\| b \|_2), C_a,$ and $C_b$ are defined as follows. The quantity $\| a \|_2$ summarizes the deviation of $a(t)$ from its posterior (pointwise empirical) mean $\hat{a}(t)$ and is defined as
	$
	\| a \|_2 = \sum_{t = 1}^T | a(t) - \hat{a}(t) |^2 
	$.
	The statistic $C_a$ is a surrogate for the number of ``change points'' in the function $a(t)$:
	\begin{equation}
	C_a = \big| \left\{ k \in \{1, \ldots, K_{\max} \} : | \log(a_k / a_{k-1}) | > .1  \right\} \big|.
	\end{equation}
	The statistics $\| b \|_2 \text{ and } C_b$ are defined analogously.

Table~\ref{tab:change_points_performance_summary} summarizes the simulation results; each algorithm is run for $2.5 \times 10^4$ iterations starting from stationarity.
While \nutsGibbs{} and \dhmc{} are comparable in their performances, as discussed earlier, \dhmc{} has the advantage that all the necessary computations can be automated within the framework of probabilistic programming languages. For a more useful comparison, therefore, we also implement the default sampling scheme used by PyMC. The algorithm updates each of the discrete parameter via a Metropolis step whose proposal distribution is a symmetric uniform integer-valued distribution with the variance calibrated to achieve an acceptance rate around 40\%. 

This example is challenging for \dhmc{} as the posterior of $\tau_k$'s are in general multi-modal conditionally on the continuous parameters. 
The complex dependency between the local shrinkage and the other parameters creates potential paths among the modes, however. It seems that \dhmc{} can exploit this complex posterior geometry efficiently and be competitive with \nutsGibbs{}. Figure~\ref{fig:change_points_posterior_samples} plots 100 \dhmc{} posterior samples of the piecewise constant volatility functions $a(t)$ and $b(t)$ to illustrate the posterior structure of the model.

\begin{table}
	\caption{Performance summary of each algorithm on the change points detection example. ``DHMC'' and ``ESS'' in the table stands for \dhmc{} and \ess{}. The term $(\pm \ldots)$ is the error estimate of our \ess{} estimators. Path length is averaged over each iteration. ``Iter time'' shows the computational time for one iteration of each algorithm relative to the fastest one.}
	\label{tab:change_points_performance_summary} 
	\centering
	\begin{tabular}{c|c|c|c|c}
		& ESS per 100 samples & ESS per minute & Path length & Iter time \\ 
		\hline
		DHMC & 13.7 ($\pm$ 1.1) & 38.7 & 87.3 & 1.03 \\ 
		\hline
		\nutsGibbs{} & 11.6 ($\pm$ 3.2) & 33.5 & 218 & 1 \\ 
		\hline
		No-U-turn / Metropolis & 6.04 ($\pm$ 1.2) & 17.5 & 217 & 1
	\end{tabular}
\end{table}

\begin{figure}
	\centering
	\subfigure{
		\includegraphics[width=.47\linewidth]{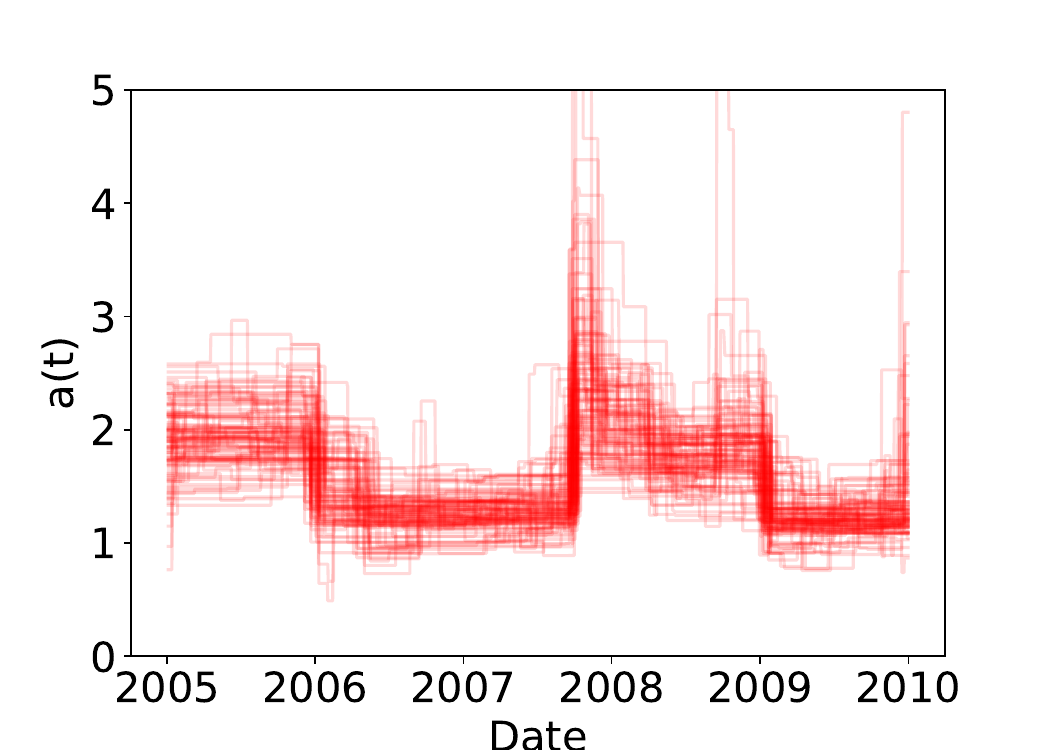} }
	\subfigure{
		\includegraphics[width=.47\linewidth]{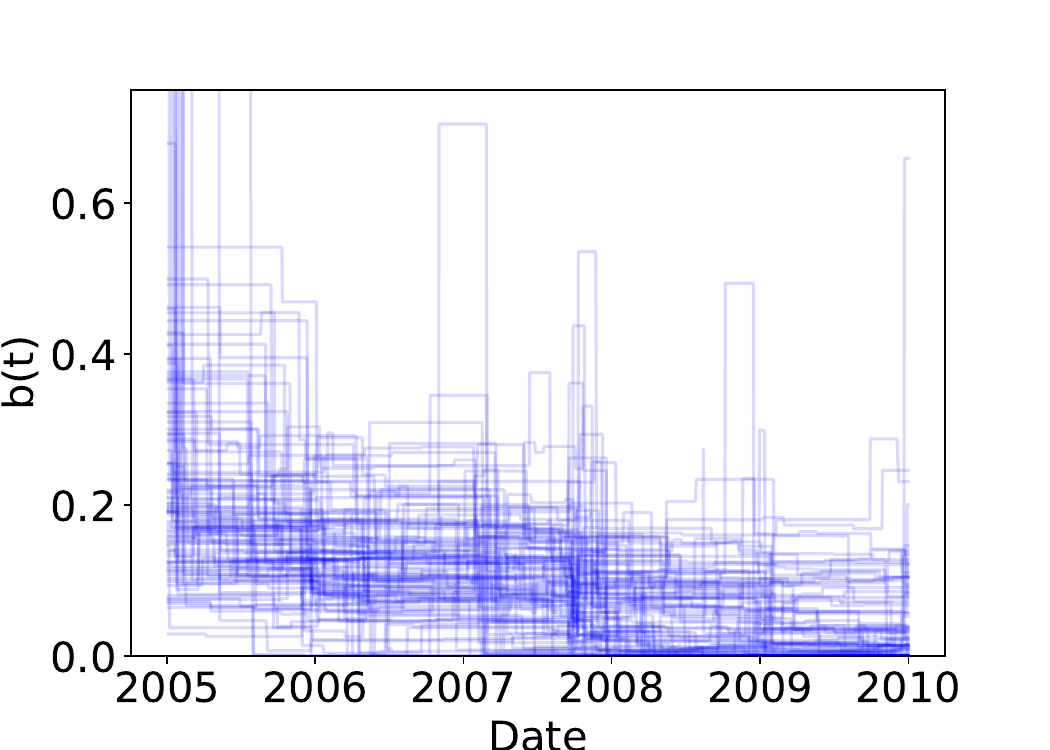} }
	\caption{Posterior samples of the piecewise constant volatility functions $a(t)$ and $b(t)$ from 100 iterations of \dhmc{}.}
	\label{fig:change_points_posterior_samples} 
\end{figure}

\section{Error analysis of \dhmc{} integrator}
\label{sec:dhmc_integrator_error_analysis}

Here we analyze the approximation error incurred by the integrator of Algorithm~\ref{alg:dhmc_integrator}. We focus on the error in Hamiltonian, the amount by which the Hamiltonian fluctuates along a numerical solution, as it determines the acceptance probability of a proposal. An error incurred by one numerical integration step $(\btheta^0, \p^0) \to (\btheta^1, \p^1)$ of stepsize $\epsilon$ is known as a \textit{local error}. Approximating the evolution $\{\btheta(0), \p(0)\} \to \{\btheta(\tau), \p(\tau)\}$ requires $L(\epsilon) = \lfloor \tau / \epsilon \rfloor$ numerical integration steps and the error incurred by the map $(\btheta^0, \p^0) \to (\btheta^L, \p^L)$ is known as a \textit{global error}. We quantify the local error of Algorithm~\ref{alg:dhmc_integrator} in Section~\ref{sec:local_error} and relate it to the global error in Section~\ref{sec:global_error}.

\subsection{Local error in Hamiltonian} 
\label{sec:local_error}

In analyzing Algorithm~\ref{alg:dhmc_integrator}, it is useful to break up the algorithm into three steps; the first (partial) update of continuous parameters, the update of discontinuous parameters, and the second update of continuous parameters. The notation $(\btheta_I^{\half}, \p_I^{\half})$ will refer to the intermediate state after the first update of continuous parameters, namely $\p_I^{\half} = \p_I^0 - \frac{\epsilon}{2} \nabla_{\btheta_I} U(\btheta_I^0, \btheta_J^0)$ and $\btheta_I^{\half} = \btheta_I^0 + \frac{\epsilon}{2} \nabla_{\p_I} K(\p_I^{\half}, \p_J^0)$ where $K(\p) = \frac{1}{2} \, \p_I^\intercal \mass_I^{-1} \p_I + \sum_{j \in J} m_j^{-1} |p_j |$ as before. The update $(\btheta_I^{0}, \p_I^{0}) \to (\btheta_I^{\half}, \p_I^{\half})$ is followed by the update $(\btheta_J^0, \p_J^0) \to (\btheta_J^1, \p_J^1)$ of discontinuous parameters, which then is followed by another continuous parameter update $(\btheta_I^{\half}, \p_I^{\half}) \to (\btheta_I^{1}, \p_I^{1})$. The exact solution is denoted by $\{\btheta(t), \p(t)\}$ with the initial condition $\{\btheta(0), \p(0)\} = (\btheta^0, \p^0)$. 

The key result in this section is Corollary~\ref{cor:local_error} below, which follows immediately from the following theorem:
\begin{theorem}
\label{thm:local_error}
	The local error in Hamiltonian incurred by Algorithm~\ref{alg:dhmc_integrator} is given by
	\begin{equation}
	\label{eq:hamiltonian_error}
	H(\btheta^1, \p^1) - H(\btheta^0, \p^0)
	= \frac{\epsilon^2}{8} \left\{ \xi\left(\btheta_I^{\half}, \btheta_J^1, \p_I^{\half}\right) - \xi\left(\btheta_I^{\half}, \btheta_J^0, \p_I^{\half}\right) \right\} + O(\epsilon^3),
	\end{equation}
	where $\xi$ is defined in terms of the Hessians $\hessi_U = \partial^2 U / \partial \btheta_I^2$ and $\hessi_K = \partial^2 K / \partial \p_I^2$ with respect to continuous parameters as 
	\begin{equation}
	\label{eq:function_for_leading_order_term}
	\begin{aligned}
	\xi(\btheta_I, \btheta_J, \p_I)
	&= \nabla_{\btheta_I}^\intercal U(\btheta_I, \btheta_J) \hessi_K(\p_I) \nabla_{\btheta_I} U(\btheta_I, \btheta_J) \\
	&\hspace{3em} - \nabla_{\p_I}^\intercal K(\p_I) \hessi_U(\btheta_I, \btheta_J) \nabla_{\p_I} K(\p_I).
	\end{aligned}
	\end{equation}
	As they are independent of $\p_J$, the derivatives of $K$ with respect to $\p_I$ are written simply as a function of $\p_I$ in the expression above.
\end{theorem}

\begin{corollary}
\label{cor:local_error}
	The local error in Hamiltonian incurred by Algorithm~\ref{alg:dhmc_integrator} is $O(\epsilon^3)$ when there is no discontinuity of $U$ along the line connecting $\btheta_J^0$ and $\btheta_J^1$. Otherwise, the local error is $O(\epsilon^2)$.
\end{corollary}

\begin{proof}[of Corollary~\ref{cor:local_error}]
	When there is no discontinuity of $U$ along the line connecting $\btheta_J^0$ and $\btheta_J^1$, the Taylor expansion of $\xi$ as defined in \eqref{eq:function_for_leading_order_term} with respect to $\btheta_J$ implies that
	\begin{equation}
	\xi\left(\btheta_I^{\half}, \btheta_J^1, \p_I^{\half}\right) - \xi\left(\btheta_I^{\half}, \btheta_J^0, \p_I^{\half}\right) 
		= O(\| \btheta_J^1 - \btheta_J^0\|)
		= O(\epsilon).
	\end{equation}
	Hence the leading order term of \eqref{eq:hamiltonian_error} becomes $O(\epsilon^3)$. 
\end{proof}

\begin{proof}[of Theorem~\ref{thm:local_error}]
	The update $(\btheta_J^{0}, \p_J^{0}) \to (\btheta_J^{1}, \p_J^{1})$ is energy-preserving by the property of the coordinate-wise integrator, so we have
	\begin{equation}
	\label{eq:hamiltonian_error_1}
	\begin{aligned}
	& H(\btheta^1, \p^1) - H(\btheta^0, \p^0) \\
	&\qquad = H(\btheta^1, \p^1) - H(\btheta_I^{\half}, \btheta_J^{1}, \p_I^{\half}, \p_J^{1}) + H(\btheta_I^{\half}, \btheta_J^{0}, \p_I^{\half}, \p_J^{0})  - H(\btheta^0, \p^0). 
	\end{aligned}
	\end{equation}
	Now let $(\btheta_I^{0}(t), \p_I^{0}(t))$ denote the solution of the differential equation
	\begin{equation}
	\label{eq:hamiltons_for_cont_param_1}
	\frac{\diff \btheta_I}{\diff t} 
	= \nabla_{\p_I} K(\p_I, \p_J^{0}), \
	\frac{\diff \p_I}{\diff t} 
	= - \nabla_{\btheta_I} U(\btheta_I, \btheta_J^{0})
	\end{equation}
	with the initial condition $\{\btheta_I^{0}(0), \p_I^{0}(0)\} = (\btheta_I^{0}, \p_I^{0})$. Similarly, let $\{\btheta_I^{\half}(t), \p_I^{\half}(t)\}$ denote the solution of the differential equation
	\begin{equation}
	\label{eq:hamiltons_for_cont_param_2}
	\frac{\diff \btheta_I}{\diff t} 
	= \nabla_{\p_I} K(\p_I, \p_J^{1}), \
	\frac{\diff \p_I}{\diff t} 
	= - \nabla_{\btheta_I} U(\btheta_I, \btheta_J^{1})
	\end{equation}
	with the initial condition $\{\btheta_I^{\half}(0), \p_I^{\half}(0)\} = (\btheta_I^{\half}, \p_I^{\half})$. By the energy-preserving property of exact Hamiltonian dynamics, \eqref{eq:hamiltonian_error_1} becomes
	\begin{equation}
	\label{eq:hamiltonian_error_2}
	\begin{aligned}
	& H(\btheta^1, \p^1) - H(\btheta^0, \p^0) \\
	&\qquad = H(\btheta^1, \p^1) - H\{\btheta_I^{\half}(\epsilon / 2), \btheta_J^{1}, \p_I^{\half}(\epsilon / 2), \p_J^{1}\} \\
	&\hspace{4em}+ H(\btheta_I^{\half}, \btheta_J^{0}, \p_I^{\half}, \p_J^{0})  - H\{\btheta_I^{0}(\epsilon / 2), \btheta_J^{0}, \p^{0}(\epsilon / 2), \p_J^{0}\}.
	\end{aligned}
	\end{equation}
	In essence, \eqref{eq:hamiltonian_error_2} shows that the error in Hamiltonian comes only from the numerical approximation errors in solving the differential equations \eqref{eq:hamiltons_for_cont_param_1} and \eqref{eq:hamiltons_for_cont_param_2}. Lemma~\ref{lem:symplectic_euler_errors} below quantifies such errors and its results can be related to the error in Hamiltonian by observing that
	\begin{equation}
	\begin{aligned}
	& H(\btheta_I^{\half}, \btheta_J^0, \p_I^{\half}, \p_J^0) - H\{\btheta_I^0(\epsilon / 2), \btheta_J^0, \p_I^0(\epsilon / 2), \p_J^0\} \\
	&\qquad = U(\btheta_I^{\half}, \btheta_J^0) - U\{\btheta_I^0(\epsilon / 2), \btheta_J^0\} + K(\p_I^{\half}, \p_J^0) - K\{\p_I^0(\epsilon / 2), \p_J^0\} \\
	&\qquad = \nabla_{\btheta_I}^\intercal U(\btheta_I^{\half}, \btheta_J^0) \big\{ \btheta_I^{\half} - \btheta_I^0(\epsilon / 2) \big\} + \nabla_{\p_I}^\intercal K(\p_I^{\half}, \p_J^0) \big\{ \p_I^{\half} - \p_I^0(\epsilon / 2) \big\} \\
	&\hspace{5em}
	+ O\big\{\| \btheta_I^{\half} - \btheta_I(\epsilon / 2) \|^2\big\}
	+ O\big\{\| \p_I^{\half} - \p_I(\epsilon / 2)  \|^2 \big\}. \\
	\end{aligned}
	\end{equation}
	Now applying \eqref{eq:symplectic_euler_error_estimate_1} of Lemma~\ref{lem:symplectic_euler_errors} with $\tileps = \epsilon / 2$, $(\btheta_I, \p_I) = (\btheta_I^{0}, \p_I^{0})$, and $(\btheta_I^*, \p_I^*) = (\btheta_I^{\half}, \p_I^{\half})$, we obtain
	\begin{equation}
	\begin{aligned}
	& H(\btheta_I^{\half}, \btheta_J^0, \p_I^{\half}, \p_J^0) - H\{\btheta_I^0(\epsilon / 2), \btheta_J^0, \p_I^0(\epsilon / 2), \p_J^0\} \\
	&\qquad = - \frac{\epsilon^2}{8} \nabla_{\btheta_I}^\intercal U(\btheta_I^{\half}, \btheta_J^0)\,  \hessi_K(\p_I^{\half}, \p_J^0) \nabla_{\btheta_I} U(\btheta_I^{\half}, \btheta_J^0) \\
	&\hspace{5em} + \frac{\epsilon^2}{8} \nabla_{\p_I}^\intercal K(\p_I^{\half}, \p_J^0) \hessi_U(\btheta_I^{\half}, \btheta_J^0) \nabla_{\p_I} K(\p_I^{\half}, \p_J^0) + O(\epsilon^3).
	\end{aligned}
	\end{equation}
	In a similar manner, it follows from \eqref{eq:symplectic_euler_error_estimate_2} of Lemma~\ref{lem:symplectic_euler_errors} that
	\begin{equation}
	\begin{aligned}
	& H(\btheta_I^{1}, \btheta_J^{1}, \p_I^{1}, \p_J^{1}) - H\{\btheta_I^{\half}(\epsilon / 2), \btheta_J^{1}, \p_I^{\half}(\epsilon / 2), \p_J^{1}\}f \\
	&\qquad = \frac{\epsilon^2}{8} \nabla_{\btheta_I}^\intercal U(\btheta_I^{1}, \btheta_J^{1})\,  \hessi_K(\p_I^{\half}, \p_J^{1}) \nabla_{\btheta_I} U(\btheta_I^{\half}, \btheta_J^{1}) \\
	&\hspace{5em} - \frac{\epsilon^2}{8} \nabla_{\p_I}^\intercal K(\p_I^{1}, \p_J^{1}) \hessi_U(\btheta_I^{\half}, \btheta_J^{1}) \nabla_{\p_I} K(\p_I^{\half}, \p_J^{1}) + O(\epsilon^3) \\
	&\qquad = \frac{\epsilon^2}{8} \nabla_{\btheta_I}^\intercal U(\btheta_I^{\half}, \btheta_J^{1})\,  \hessi_K(\p_I^{\half}, \p_J^{1}) \nabla_{\btheta_I} U(\btheta_I^{\half}, \btheta_J^{1}) \\
	&\hspace{5em} - \frac{\epsilon^2}{8} \nabla_{\p_I}^\intercal K(\p_I^{\half}, \p_J^{1}) \hessi_U(\btheta_I^{\half}, \btheta_J^{1}) \nabla_{\p_I} K(\p_I^{\half}, \p_J^{1}) + O(\epsilon^3).
	\end{aligned}
	\end{equation}
	The result \eqref{eq:hamiltonian_error} now follows by simply noting that the derivatives of $K$ with respect to $\p_I$ are independent of $\p_J$. 
\end{proof}

\begin{lemma}
\label{lem:symplectic_euler_errors}
	For $(\btheta_J, \p_J)$ fixed, let $\{\btheta_I(t), \p_I(t)\}$ denote the solution of the differential equation
	\begin{equation}
	\frac{\diff \btheta_I}{\diff t} 
	= \nabla_{\p_I} K(\p_I, \p_J), \
	\frac{\diff \p_I}{\diff t} 
	= - \nabla_{\btheta_I} U(\btheta_I, \btheta_J)
	\end{equation}
	with the initial condition $(\btheta_I(0), \p_I(0)) = (\btheta_I, \p_I)$. The approximation error of the numerical scheme 
	\begin{equation}
	\label{eq:symplectic_euler_1}
	\btheta_I^* = \btheta_I + \tileps \nabla_{\p_I} K(\p_I^*, \p_J), \
	\p_I^* = \p_I - \tileps \nabla_{\btheta_I} U(\btheta_I, \btheta_J)
	\end{equation}
	satisfies 
	\begin{equation}
	\label{eq:symplectic_euler_error_estimate_1}
	\begin{aligned}
	\btheta_I^* - \btheta_I(\tileps) 
	= - \frac{\tileps^2}{2} \hessi_K(\p_I^*, \p_J) \nabla_{\btheta_I} U(\btheta_I^*, \btheta_J) + O(\tileps^3) \\
	\p_I^* - \p_I(\tileps) 
	= \frac{\tileps^2}{2} \hessi_U(\btheta_I^*, \btheta_J) \nabla_{\p_I} K(\p_I^*, \p_J) + O(\tileps^3) 
	\end{aligned}
	\end{equation}
	where $\hessi_U = \partial^2 U / \partial \btheta_I^2$ and $\hessi_K = \partial^2 K / \partial \p_I^2$ are the Hessians of $U$ and $K$ with respect to $\btheta_I$ and $\p_I$. Similarly, the approximation error of the numerical scheme
	\begin{equation}
	\btheta_I^* = \btheta_I + \tileps \nabla_{\p_I} K(\p_I, \p_J), \
	\p_I^* = \p_I - \tileps \nabla_{\btheta_I} U(\btheta_I^*, \btheta_J)
	\end{equation}
	satisfies 
	\begin{equation}
	\label{eq:symplectic_euler_error_estimate_2}
	\begin{aligned}
	\btheta_I^* - \btheta_I(\tileps) 
	= \frac{\tileps^2}{2} \hessi_K(\p_I, \p_J) \nabla_{\btheta_I} U(\btheta_I, \btheta_J) + O(\tileps^3) \\
	\p_I^* - \p_I(\tileps) 
	= - \frac{\tileps^2}{2} \hessi_U(\btheta_I, \btheta_J) \nabla_{\p_I} K(\p_I, \p_J) + O(\tileps^3).
	\end{aligned}
	\end{equation}
\end{lemma}

\begin{proof}
	The proofs of \eqref{eq:symplectic_euler_error_estimate_1} and \eqref{eq:symplectic_euler_error_estimate_2} are very similar, so we focus on the derivations of \eqref{eq:symplectic_euler_error_estimate_1}. Taylor expansion of $\btheta_I(t)$ yields
	\begin{equation}
	\label{eq:exact_solution_taylor}
	\begin{aligned}
	\btheta_I(\tileps) - \btheta_I
		&= \tileps \frac{\diff \btheta}{\diff t} + \frac{\tileps^2}{2} \frac{\diff^2 \btheta}{\diff t^2} + O(\tileps^3) \\
		&= \tileps \nabla_{\p_I} K(\p_I, \p_J) - \frac{\tileps^2}{2} \hessi_K(\p_I, \p_J) \nabla_{\btheta_I} U(\btheta_I, \btheta_J) + O(\tileps^3).
	\end{aligned}
	\end{equation}
	On the other hand, Taylor expansion of $\nabla_{\p_I} K(\p_I^*, \p_J)$ in the first variable yields  
	\begin{equation}
	\label{eq:numerical_solution_taylor}
	\begin{aligned}
	\btheta_I^* - \btheta_I
		&= \tileps \nabla_{\p_I} K(\p_I, \p_J) + \tileps \, \hessi_K(\p_I, \p_J) (\p_I^* - \p_I) + \tileps \, O(\| \p_I^* - \p_I \|^2) \\
		&= \tileps \nabla_{\p_I} K(\p_I, \p_J) - \tileps^2 \hessi_K(\p_I, \p_J) \nabla_{\btheta_I} U(\btheta_I, \btheta_J) + O(\tileps^3).
	\end{aligned}
	\end{equation}
	Subtracting \eqref{eq:exact_solution_taylor} from \eqref{eq:numerical_solution_taylor}, we obtain
	\begin{equation}
	\begin{aligned}
	\btheta_I^* - \btheta_I(\tileps)
		&= - \frac{\tileps^2}{2} \hessi_K(\p_I, \p_J) \nabla_{\btheta_I} U(\btheta_I, \btheta_J) + O(\tileps^3) \\
		&= - \frac{\tileps^2}{2} \hessi_K(\p_I^*, \p_J) \nabla_{\btheta_I} U(\btheta_I^*, \btheta_J) + O(\tileps^3),
	\end{aligned}
	\end{equation}
	where the second equality again follows from a Taylor expansion applied to the leading order term. The error estimate for the momentum variable is similar; the Taylor expansion of $\p_I(t)$ gives 
	\begin{equation}
	\label{eq:momentum_exact_solution_taylor}
	\p_I(\tileps) - \p_I
		= - \tileps \nabla_{\btheta_I} U(\btheta_I, \btheta_J) - \frac{\tileps^2}{2} \hessi_U(\btheta_I, \btheta_J) \nabla_{\p_I} K(\p_I, \p_J) + O(\tileps^3).
	\end{equation}
	Subtracting \eqref{eq:momentum_exact_solution_taylor} from \eqref{eq:symplectic_euler_1}, we obtain
	\begin{equation}
	\begin{aligned}
	\p_I^* - \p_I(\tileps) 
	&= \frac{\tileps^2}{2} \hessi_U(\btheta_I, \btheta_J) \nabla_{\p_I} K(\p_I, \p_J) + O(\tileps^3) \\
	&= \frac{\tileps^2}{2} \hessi_U(\btheta_I^*, \btheta_J) \nabla_{\p_I} K(\p_I^*, \p_J) + O(\tileps^3). 
	\end{aligned}
	\hspace*{2em} \mbox{\QEDlogo}
	\end{equation}
\end{proof}

\subsection{Global error in Hamiltonian}
\label{sec:global_error}

Theorem~\ref{thm:global_convergence} below establishes the global error in Hamiltonian to be $O(\epsilon^2)$. 
For its proof, we recall that Algorithm~\ref{alg:dhmc_integrator} is designed under the assumption that the parameter space has a partition $\mathbb{R}^{|I|} \times \mathbb{R}^{|J|} = \cup_k \mathbb{R}^{|I|} \times \Omega_k$ such that $U(\btheta)$ is smooth on $\mathbb{R}^{|I|} \times \Omega_k$ for each $k$. Below, in relating the local error to the global one, we make the dependence of a numerical solution on a stepsize $\epsilon$ explicit and denote the value of a numerical solution after $\ell$ steps by $(\btheta_\epsilon^{\ell}, \p_\epsilon^{\ell})$. 

\begin{theorem}
\label{thm:global_convergence}
	Suppose that each $\Omega_k$ is rectangular so that its boundary consists of planes perpendicular to one of the coordinates of $\btheta_J$. Then the global error $H(\btheta_\epsilon^{L}, \p_\epsilon^{L}) - H(\btheta^{0}, \p^{0})$, with $L = L(\epsilon) = \lfloor \tau / \epsilon \rfloor$, incurred by Algorithm~\ref{alg:dhmc_integrator} is of order $O\left( \epsilon^2 D  \right)$ where $D$ is the number of discontinuities in $U$ encountered along the trajectory $\{\btheta(t), 0 \leq t \leq \tau \}$.
\end{theorem}
\noindent The assumption stated in Theorem~\ref{thm:global_convergence} is required for our proof of Theorem~\ref{thm:strong_convergence} and is satisfied whenever the discontinuous target $\pi$ is obtained by the embedding of discrete parameters described in Section~\ref{sec:embedding}. We however believe the order of the global error remains unchanged under more general conditions.

\begin{proof} 
	The global error is given as a sum of the local errors:
	\begin{equation}
	\label{eq:global_as_sum_of_local}
	H(\btheta_\epsilon^{L}, \p_\epsilon^{L}) - H(\btheta^{0}, \p^{0})
		= \sum_{\ell = 1}^{L} \left\{ H(\btheta_\epsilon^{\ell}, \p_\epsilon^{\ell})  - H(\btheta_\epsilon^{\ell - 1}, \p_\epsilon^{\ell - 1}) \right\}
	\end{equation}
	Let $D(\epsilon)$ denote the size of the set $\mathcal{D}_\epsilon$ as defined below:
	\begin{equation}
	\begin{aligned}
	\mathcal{D}_\epsilon 
		&= \Big\{ \ell \in \{1, \ldots, L \}: \\
		&\hspace{3em} \text{$\btheta_{\epsilon, J}^{\ell}$ and $\btheta_{\epsilon, J}^{\ell - 1}$ belong to two separate regions of the partition $\Omega_k$'s} \Big\}
	\end{aligned}
	\end{equation}
	By the result of Corollary~\ref{cor:local_error}, we know that the local error is $O(\epsilon^2)$ if $\ell \in \mathcal{D}_\epsilon$ and $O(\epsilon^3)$ otherwise. Therefore, \eqref{eq:global_as_sum_of_local} is a sum of $D(\epsilon) $ terms of $O(\epsilon^2)$ errors and $L(\epsilon) - D(\epsilon)$ terms of $O(\epsilon^3)$ errors, yielding the global error of $O\!\left\{D(\epsilon) \epsilon^2\right\}$. To complete the proof, it follows from Theorem~\ref{thm:strong_convergence} that $D(\epsilon)$ as $\epsilon \to 0$ converges to the number of discontinuities in $U$ encountered along the trajectory $\{\btheta(t), 0 \leq t \leq \tau \}$. 
\end{proof}

\begin{theorem}
\label{thm:strong_convergence}
	Under the assumption of Theorem~\ref{thm:global_convergence}, we have
	\begin{equation}
	\sup_{\ell = 1, \ldots, L(\epsilon)} \big\| \{\btheta(\ell \epsilon), \p(\ell \epsilon)\} - (\btheta_\epsilon^{\ell}, \p_\epsilon^{\ell}) \big\| = O(\epsilon).
	\end{equation}
\end{theorem}

\begin{proof}
	First note that the trajectory of Hamiltonian dynamics corresponding to the kinetic energy \eqref{eq:mixed_kinetic_energy} can be partitioned into $\widetilde{D}$ segments $\{ \btheta(t): t_m < t < t_{m + 1}\}_m$ for $0 = t_0 < t_1 < \ldots < t_{\widetilde{D}} = \tau$ so that on each segment $\diff \btheta_J / \diff t = \bm{m}_J^{- 1} \odot \sign(\p_J)$ is constant. 
	
	The numerical solution approximates the exact solution $\btheta(t) \to \btheta(t + \ell \epsilon)$ up to an error of $O(\epsilon^2)$ for any $\ell$ provided that $\btheta(t)$ and $\btheta(t + \ell \epsilon)$ belong to the same segment $\{ \btheta(t): t_m < t < t_{m + 1}\}$. This is for the following reason. For all sufficiently small $\epsilon$, the coordinate-wise updates of discontinuous parameters  yield the exact solution to 
	\begin{equation}
	\begin{aligned}
	&\frac{\diff \btheta_J}{\diff t}
	= \bm{m}_J^{-1} \odot \textrm{sign}(\p_J), \quad
	\frac{\diff \p_J}{\diff t} = - \nabla_{\btheta_J} U(\btheta), \quad
	\frac{\diff \btheta_{\minus J}}{\diff t} = \frac{\diff \p_{\minus J}}{\diff t} = \bm{0}
	\end{aligned}
	\end{equation}
	provided no sign change in $\p_J$ is encountered. 
	 In this case, Algorithm~\ref{alg:dhmc_integrator} coincides with a symmetric splitting of Hamilton's equation in which the individual components are solved exactly and hence the numerical approximation of $\btheta(t) \to \btheta(t + \epsilon)$ locally agrees with the exact solution up to an error of $O(\epsilon^3)$ \citep{leimkuhler04}. 
	
	Now consider the case when $\btheta(t)$ and $\btheta(t + \epsilon)$ do not belong to the same segment. In this case, the coordinate-wise integrator approximates the change in $\diff \btheta_J / \diff t $ through the momentum flip $p_j \to - p_j$ for an appropriate $j$ with $\theta_j$ held fixed. This may or may not be caused by a discontinuity in $U$ along the path $\{\btheta(s): t < s < t + \epsilon\}$. When there is no discontinuity, the approximation is always accurate up to an error of $O(\epsilon)$. When there is a discontinuity, our assumption on the boundaries of $\Omega_k$'s guarantees the numerical approximation error to be $O(\epsilon)$.
	
	To summarize, we have shown that the total accumulated error is $O(\epsilon^2)$ while the solution stays within the same segment $\{ \btheta(t): t_m < t < t_{m + 1}\}_m$ and then an additional error of $O(\epsilon)$ is incurred when crossing from one segment to another. Since the solution trajectory consists of $\widetilde{D}$ such segments, the total accumulated error is $O\{\widetilde{D} (\epsilon + \epsilon^2)\} = O(\widetilde{D} \epsilon)$. 
\end{proof}

\end{document}